\documentclass[showpacs,pra,twocolumn,notitlepage]{revtex4-1}

\usepackage[dvips]{graphicx} 
\usepackage{amsmath,amssymb,amsthm,mathrsfs,amsfonts,dsfont}
\usepackage{enumerate}
\usepackage{epsfig}
\usepackage{subfigure}
\usepackage{xcolor}
\usepackage{amsthm}
\usepackage{physics}
\usepackage{multirow}    
\usepackage{array}     
\usepackage{comment}
\usepackage{ulem}
\usepackage{enumitem} 
\usepackage{epstopdf}

\graphicspath{{./figure/}}

\newtheorem{theorem}{Theorem}
\newtheorem{lemma}{Lemma}	
\newtheorem{corollary}{Corollary}

\newtheorem{definition}{Definition}

\newtheorem{observation}{Observation}

\begin{document}
\title{Symmetry-protected privacy: beating the rate-distance linear bound over a noisy channel}

\author{Pei Zeng}
\affiliation{Center for Quantum Information, Institute for Interdisciplinary Information Sciences, Tsinghua University, Beijing 100084, China}
\author{Weijie Wu}
\affiliation{Center for Quantum Information, Institute for Interdisciplinary Information Sciences, Tsinghua University, Beijing 100084, China}
\author{Xiongfeng Ma}
\email{xma@tsinghua.edu.cn}
\affiliation{Center for Quantum Information, Institute for Interdisciplinary Information Sciences, Tsinghua University, Beijing 100084, China}

\begin{abstract}
There are two main factors limiting the performance of quantum key distribution --- channel transmission loss and noise. Previously, a linear bound was believed to put an upper limit on the rate-transmittance performance. Remarkably, the recently proposed twin-field and phase-matching quantum key distribution schemes have been proven to overcome the linear bound. In practice, due to the intractable phase fluctuation of optical signals in transmission, these schemes suffer from large error rates, which renders the experimental realization extremely challenging. Here, we close this gap by proving the security based on a different principle --- encoding symmetry. With the symmetry-based security proof, we can decouple the privacy from the channel disturbance, and eventually remove the limitation of secure key distribution on bit error rates. That is, the phase-matching scheme can yield positive key rates even with high bit error rates up to 50\%. In simulation, with typical experimental parameters, the key rate is able to break the linear bound with an error rate of 13\%. Meanwhile, we provide a finite-data size analysis for the scheme, which can break the bound with a reasonable data size of $10^{12}$. Encouraged by high loss- and error-tolerance, we expect the approach based on symmetry-protected privacy will provide a new insight to the security of quantum key distribution.
\end{abstract}

\maketitle

\section{Introduction}
Quantum key distribution (QKD) offers information-theoretically secure means to distribute private keys between distant parties by harnessing the laws of quantum mechanics \cite{Bennett1984Quantum,ekert1991quantum}. The commercialization and application of QKD raise requirements in both impregnable security and outstanding performance.

The security, as the cornerstone of QKD, has been proven theoretically at the end of last century \cite{Mayers2001Unconditional,Lo1999Unconditional,Shor2000Simple} on the protocol level, while rigorous definition \cite{Ben2005composable,Renner2005universally} and strict finite-size analysis \cite{renner2008security,Fung2010Practical} have been provided later. The security of QKD is based on the idea that information gain means disturbance. That is, any eavesdropper's attempt of learning the keys would inevitably introduce disturbance to the quantum states. To characterize the information leakage, the disturbance in the channel is monitored in real time. In practice, the physical devices used in practical implementations often deviate from the assumed theoretical models \cite{gottesman2004security}, resulting in various loopholes and corresponding attacks \cite{gisin2006trojan,Qi2007Time}. In 2012, measurement-device-independent quantum key distribution (MDI-QKD) has been presented \cite{Lo2012Measurement}, which removes the theoretical assumptions on measurement devices in security analysis and hence closes all the detection loopholes.

The performance of QKD, on the other hand, characterized by the key generation rate with respect to the communication distance, reflects its value in commercial cryptographic task. Under the circumstances that quantum repeaters \cite{Briegel1998Repeater,duan2001long,azuma2015all}, as the ultimate solution to extend quantum communication against losses, are currently infeasible, the linear key rate-transmittance bound \cite{Pirandola2017Fundamental} was widely believed to hold for all the point-to-point QKD schemes without repeaters. For the commonly used telecom fiber channel, the transmittance decreases exponentially with the transmission distance, which puts an upper limit on quantum transmission distance. Interestingly, the recent work of twin-field QKD open the new possibility of phase-encoding MDI-QKD protocol to break the linear key rate bound \cite{Lucamarini2018TF}.
A follow-up work, named phase-matching quantum key distribution (PM-QKD) has been proposed \cite{ma2018phase,Lin2018simple}, which has been rigorously shown to be able to beat the linear bound even with statistical fluctuations \cite{maeda2019repeaterless}. The twin-field-like MDI-QKD is currently a heated topic \cite{Tamaki2018information,Wang2018Sending,cui2019twin,Curty2018simple}.

In PM-QKD, only one basis is adopted for key generation and parameter estimation, which is distinct from the former BB84-type protocol while shares some similarities with the Bennett-1992 protocol \cite{Bennett1992Quantum}. Essentially, the PM-QKD protocol can be viewed as an MDI version of the Bennett-1992 protocol \cite{Ferenczi2013}. To understand the security of PM-QKD, we would resort to the non-orthogonality of the encoded state with $0/\pi$ encoding. However, the analysis based on non-orthogonality usually cannot tolerate high channel losses.

In this work, for a general PM-QKD model, we establish a connection between the encoding symmetry and privacy, which serves as a new viewpoint of QKD security other than the conventional basis complementarity \cite{koashi2009simple}. In this symmetry-based security proof, we first define symmetric states given certain encoding operations, then explore the realistic construction of symmetric states, and finally propose efficient methods to estimate the ratio of detection caused by symmetric states. For PM-QKD, the symmetric state can promise perfect privacy, i.e., with no information leakage. As a result, the amount of information leakage only depends on the state producing by the source, and is irrelevant of the channel condition. A similar phenomenon is also observed in the round-robin differential-phase-shifting protocol \cite{sasaki2014practical}. The symmetry-based security proof allows higher error tolerance compared with the original BB84 protocol. Furthermore, we complete finite-size analysis with an improved decoy-state method \cite{Hwang2003Decoy,Lo2005Decoy,Wang2005Decoy}.


\section{Encoding symmetry and perfect privacy}
To show the close relationship between encoding symmetry and privacy, we introduce an idealized scenario, named symmetric-encoding QKD. As is shown in Fig.~\ref{fig:ProSym}, during each run, Alice and Bob start with a pre-shared bipartite state $\rho_{AB}$, where the systems held by Alice and Bob are denoted as $A$ and $B$, respectively. They generate random bits $\kappa_a$ and $\kappa_b$ independently and apply $U(\kappa_{a(b)})\equiv U^{\kappa_{a(b)}}$ to their subsystem $A$ and $B$ separately, where $U^2 = I$. Then, the modulated state, denoted as $\rho'_{AB}(\kappa_a, \kappa_b)$, is sent to an untrusted party Eve, who measures the joint state, aiming to discriminate whether $\kappa_a = \kappa_b$ or $\kappa_a \neq \kappa_b$, and announces the detection result. In each round, Alice and Bob encode $2$-bit information, $\kappa_a, \kappa_b$, into the state $\rho_{AB}$. The encoded state $\rho'_{AB}(\kappa_a, \kappa_b)$ can be written as
\begin{equation}
\rho'_{AB}(\kappa_a, \kappa_b) = [U_A(\kappa_a)\otimes U_B(\kappa_b)]\rho_{AB}[U_A(\kappa_a)\otimes U_B(\kappa_b)]^\dagger.
\end{equation}
After many rounds, Alice and Bob generate random bit strings $K_a$ and $K_b$, respectively. With the assistance of Eve's announcement and classical error correction, Bob reconciles his bit string $K_b$ to $K_a$. They then perform privacy amplification on $K_a$ to extract a private key.

\begin{figure}[tbhp]
\centering
\includegraphics[width=6.5cm]{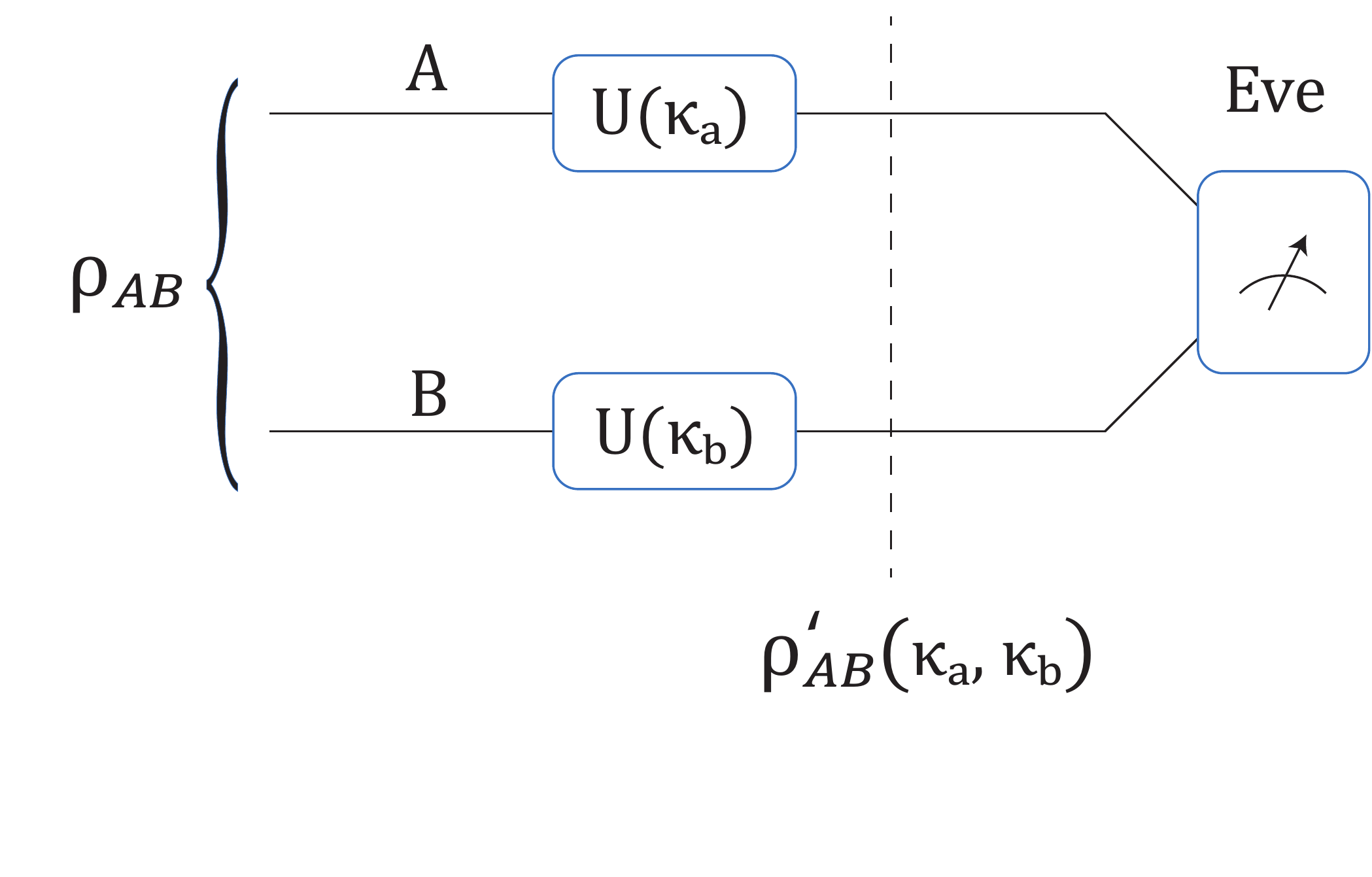}
\caption{Schematic diagram of the symmetric-encoding protocol.}
\label{fig:ProSym}
\end{figure}

Now, let us focus on a symmetric case, where the pre-shared state $\rho_{AB}$ keeps invariant under the transformation of encoding operation $U_A \otimes U_B$. That is, a pure state $\ket{\psi}_{AB}$ is invariant under the encoding operation if
\begin{equation}
\ket{\psi}_{AB} = U_A \otimes U_B \ket{\psi}_{AB}.
\end{equation}
Then, $\ket{\psi}_{AB}$ is an eigenstate of $U_A \otimes U_B$. Since $(U_A\otimes U_B)^2 = I$, the eigenvalue of $U_A\otimes U_B$ is either $+1$ or $-1$. We name the eigenvalue $+1$ subspace of $U_A\otimes U_B$ as the even space $\mathcal{H}^{even}$ and the eigenvalue $-1$ subspace as the odd space $\mathcal{H}^{odd}$. The states lie in $\mathcal{H}^{even}$ are called even states and the states lie in $\mathcal{H}^{odd}$ are called odd states. Together they are called parity states. Obviously, a mixture of odd (even) states is still an odd (even) state.

For a generic mixture of pure parity states, $\rho_{AB} = \sum_{i} p_i \ket{\psi_i}\bra{\psi_i}$, where all the components $\{\ket{\psi_i}\}$ are parity states, it keeps invariant under $U_A \otimes U_B$,
\begin{equation} \label{eq:equalencoding}
\begin{aligned}
\rho'_{AB}(0,0) = \rho'_{AB}(1,1), \\
\rho'_{AB}(0,1) = \rho'_{AB}(1,0).
\end{aligned}
\end{equation}
The raw key bit information $\kappa_a$ is ``hidden'' on the encoded state $\rho'_{AB}(\kappa_a, \kappa_b)$. However, when Eve holds the purification of $\rho_{AB}$, she may still learn $\kappa_a$ from the encoded state. Without loss of generality, we consider a purification of $\rho_{AB}$, $\ket{\Psi}_{ABC} = \sum_{i} \sqrt{p_i} \ket{\psi_i}_{AB} \ket{i}_C$, where system $C$ is held by Eve. Under the encoding operation $U_A\otimes U_B$, all the odd and even state components of $\{\ket{\psi_i}\}$ will gain a factor $-1$ and $1$, respectively. If there are only odd or even components in $\rho_{AB}$, the purified state $\ket{\Psi}_{ABC}$ will keep unchanged under the encoding operation. On the other hand, the coexistance of odd and even components in $\rho_{AB}$ will lead to a change of relative phase in $\ket{\Psi}_{ABC}$, allowing Eve to discriminate $\rho'_{AB}(0,0)$ and $\rho'_{AB}(1,1)$.

To make the observations above rigorous, in Appendix \ref{Sc:Proof}, we analyze the security of symmetric-encoding QKD with a standard phase-error-correction approach \cite{Lo1999Unconditional,Shor2000Simple,koashi2009simple}. When the input state $\rho_{AB}$ is an (even or odd) parity state, the symmetric-encoding QKD protocol shown in Fig.~\ref{fig:ProSym}, is perfectly private, which is reflected by a zero phase error rate $E^{(ph)}=0$. Moreover, if $\rho_{AB}$ is a mixture of even and odd state,
\begin{equation} \label{eq:mixoddeven}
\begin{aligned}
\rho_{AB} = p_{odd} \rho_{odd} + p_{even} \rho_{even},
\end{aligned}
\end{equation}
then we can estimate the ratios of odd and even components causing effective detection,
\begin{equation}
\begin{aligned}
q_{odd} &= p_{odd} \frac{Y_{odd}}{Q}, \\
q_{even} &= p_{even} \frac{Y_{even}}{Q},
\end{aligned}
\end{equation}
where $Y_{odd}$, $Y_{even}$, and $Q$ are the successful detection probability of $\rho_{odd}$, $\rho_{even}$, and $\rho_{AB}$, respectively. The phase error rate is $E^{(ph)} = q_{even}$. Therefore, the key rate is given by
\begin{equation} \label{eq:keyrate}
\begin{aligned}
r &= 1-H(E)-H(q_{even})
\end{aligned}
\end{equation}
where $E$ is the quantum bit error rate. Here, we can see that as long as the final state post-selected by successful detection is close to a parity state (either even or odd), the information leakage can be bounded.

\section{Phase-matching quantum key distribution} \label{Sc:Protocol}
The problem for implementing the symmetric-encoding protocol is that both the parity states $\rho_{odd}$ and $\rho_{even}$ are usually nonlocal. That is, they cannot be obtained by Alice and Bob via independent local state preparation. Thus, they will inevitably prepare a mixture of even and odd parity states in practice. The PM-QKD protocol can be regarded as a realization of symmetric-encoding QKD, where Alice and Bob construct the parity state input $\rho_{AB}$ using two optical modes based on independent laser sources. To construct parity state $\rho_{AB}$ from experimentally accessible coherent states, a natural way is to perform simultaneous $0/\pi$ phase-randomization on two coherent states $\ket{\sqrt{\mu_a}}_A, \ket{\sqrt{\mu_b}}_B$ to decouple the odd and even photon components.

For the convenience of parameter estimation, we consider the PM-QKD protocol where Alice and Bob randomize the phase $\phi$ on $\ket{\sqrt{\mu_a }e^{i\phi}}_A, \ket{\sqrt{\mu_b }e^{i\phi}}_B$ continuously so that the photon number components $\{\ket{m,n}\}_{m+n=k}$ are decoupled,
\begin{equation}
\rho_{AB} = \sum_{k=0}^{\infty} p_k \rho_k,
\end{equation}
where $\rho_k$ a pure parity state since it is a Fock state. The overall phase error of PM-QKD is given by
\begin{equation} \label{eq:EX}
q_{even} = 1 - \sum_{k} q_{2k+1} 
\end{equation}
where $q_k$ is the fraction of detection when Alice and Bob send out $k$-photon signals,
\begin{equation} \label{eq:qk}
q_k = P_{\mu_t}(k)\dfrac{Y_k}{Q_{\mu_t}}.
\end{equation}
Here $\mu_t = \mu_a + \mu_b$; $Y_k$ is the yield of $k$-photon component; $Q_{\mu_t}$ is the overall gain, i.e., the successful detection probability when Alice and Bob send out coherent lights with intensities of $\mu_a$ and $\mu_b$, respectively; and $P_{\mu_t}(k) = e^{-2\mu} (2\mu)^k/(k!)$ is the Poisson distribution. In order to estimate the information leakage, we only need to estimate the fraction of odd state detections.

The simultaneous phase-randomization is also nonlocal. To achieve this in practice, Alice and Bob first randomize the phase independently, and then post-select the pulses with the same random phase by phase announcement and sifting \cite{Ma2012Alternative}. The overall privacy, characterized by the overall phase error rate, will not change after the random phase announcement \cite{ma2018phase,maeda2019repeaterless}, which indicates that simultaneous phase randomization can be replaced by independent phase randomization and post-selection.

In each turn, Alice and Bob each generates a random phase $\phi_{a(b)}$ and a random key bit $\kappa_{a(b)}$. They then modulate their coherent pulse $\ket{\sqrt{\mu_{a(b)}}}$ by a phase $(\phi_{a(b)} + \pi\kappa_{a(b)})$. After Eve announces the detection result, they announce the random phases $\phi_{a(b)}$ to group the signals with the same random phase difference. From the detection result, they estimate the information leakage and extract key from the encoding bits. In practice, the continuous phase randomization can be replaced with discrete randomization \cite{Cao2015Discrete}. The detailed analysis of PM-QKD with phase post-selection and discrete phase randomization is presented in Appendix \ref{Sc:SecurePM}.

In discrete-phase encoding, Alice and Bob randomly pick up one of the $D$ phases equally distributed in $[0,2\pi)$. They announce the discrete random phase $\phi_a = \frac{2\pi}{D} j_a$ and $\phi_b = \frac{2\pi}{D} j_b$ with the indexes $j_{a(b)} = 0,1,...,D-1$. Based on the random phase difference, they group the signals by $j_s = (j_b - j_a) \text{ mod } \frac{D}{2}$. For example, when $D=16$, the signals with $j_b - j_a = 1$ and $9$ are in the same group. After grouping, there is $\frac{D}{2}$ groups with the label of $j_s = 0,1,...,\frac{D}{2}-1$. In the ideal case, the signals with $j_s = 0$ are the signals with matched phase. For the signals with $j_s \neq 0$, there is an intrinsic mismatched phase $\phi_\delta$. It is conservative to regard $\phi_\delta$ being caused by Eve in security analysis. For each group of data, the information leakage can be bounded by $q_{even}$ in Eq.~\eqref{eq:EX}, regardless of Eve's measurement setting or the bit error rates.

Thanks to this decoupled relationship between privacy and channel disturbance, we can improve the post-processing step by utilizing the unaligned data with $j_s \neq 0$. Alice and Bob first reconcile their sifted raw key bits $K_a$ and $K_b$ with for each group $j_s$ separately. If the error rate in a group of data is too large, they can simply discard that group. Denote the group set $J$ to be the set of remaining phase group indexes $\{j_s\}$. That is, if $j_s \in J$, then the phase group $j_s$ is kept for key generation. They then estimate the even photon fraction $q_{even}$ for all the remaining data and perform privacy amplification. Note that $q_{even}$ is the same for different data group $j_s$. The overall key rate of PM-QKD, taking the phase sifting and loss into account, is given by
\begin{equation} \label{eq:groupkey}
R = \dfrac{2 Q_{\mu}}{D} \sum_{j_s\in J} \left[ 1 - H(q_{even}) - f H(E_{j_s}) \right],
\end{equation}
where $f$ is the error correction efficiency, and $E_{j_s}$ is the bit error rate of phase group $j_s$ with $\mu_{i_a} = \mu_{i_b} = {\mu}/{2}$. In the experiment, all the parameters in Eq.~\eqref{eq:groupkey} can be directly obtained except for $q_{even}$, which needs to be bounded by the decoy-state method. It has been shown in literature that with the infinite decoy-state method \cite{Lo2005Decoy}, all the parameters, including $q_{even}$ can be estimated accurately \cite{ma2018phase,Lin2018simple}.

\section{Robustness to loss and noise}
To test the capability to tolerate high noises, we simulate the performance of PM-QKD in the asymptotic limit, under different level of misalignment errors, with $D=16$ discrete phases. In simulation model, there are three major error source: the system misalignment error $e_d$, caused by phase fluctuation and system misalignment; the background error, caused by dark counts $p_d$; and the mismatch error $e_\Delta(j_s)$ caused by the intrinsic mismatch for different phase groups,
\begin{equation}
e_\Delta(j_s) =
\begin{cases}
\sin^2(\frac{\pi j_s}{D}), 0 \leq j_s \leq \frac{D}{4}, \\
\sin^2(\frac{\pi}{2} - \frac{\pi j_s}{D}), \frac{D}{4} < j_s \leq \frac{D}{2}-1, \\
\end{cases}
\end{equation}
which is related to the index deviation $0\le j_s\le D/2-1$. 
The overall bit error rate $E_\mu^{(j_s)}$ is then given by
\begin{equation}
E_\mu^{(j_s)} = \min\left\{ [p_d + \eta\mu (e_\Delta(j_s) + e_d)]\frac{e^{-\eta\mu}}{Q_\mu}, 0.5\right\},
\end{equation}
where $\eta$ is the transmittance.

In practice, the phase drift caused by lasers and fiber links will degrade the performance of PM-QKD. To avoid the effect caused by phase drift, one demanding way is to introduce active feedback and phase locking. Another enhancement is to introduce the phase post-compensation method \cite{ma2018phase}, where Alice and Bob can estimate the phase drift $\phi_\delta$ by strong light pulses, record it and take it into account in the sifting step. Suppose the phase drift is slow, then Alice and Bob is able to figure out the closest discrete phase $\phi_j \equiv \frac{2\pi}{D} j$ to $\phi_\delta$. Denote the index of the closest discrete phase as $j_\delta$. During the sifting step, Alice and Bob modify the definition of $j_s$ to $j_s = j_a - j_b - j_\delta$. In this sense, Alice and Bob ``post-compensate'' the effect caused by phase drift.

With the parameters given in Table~\ref{tab:para}, we simulate the asymptotic key rate performance under the setting of system misalignment error rates $e_d$ of $1\%, 5\%, 9\%$ and $13\%$. Note that, the normal symmetric BB84 protocol cannot yield any positive key rate under the misalignement error $e_d\ge 11\%$. Surprisingly, from Fig.~\ref{fig:InfKey}, one can see, even if $e_d=13\%$, the key rate of PM-QKD is still able to surpasses the linear bound when $l>330$ km. This illustrates the robustness of PM-QKD against both noisy and lossy channel. In an extreme case when the source light intensity $\mu \to 0$ and the dark count of detector $p_d = 0$, the single photon fraction among all the detected signal $q_1 \to 1$ according to Eq.~\eqref{eq:qk}, hence the phase error rate $E^{(ph)} \to 0$. In this case, the key rate of PM-QKD $R>0$ even when the bit error $E^{(j_s)}_\mu$ is close to $50\%$.

\begin{figure}[htbp]
\resizebox{8cm}{!}{\includegraphics[scale=1]{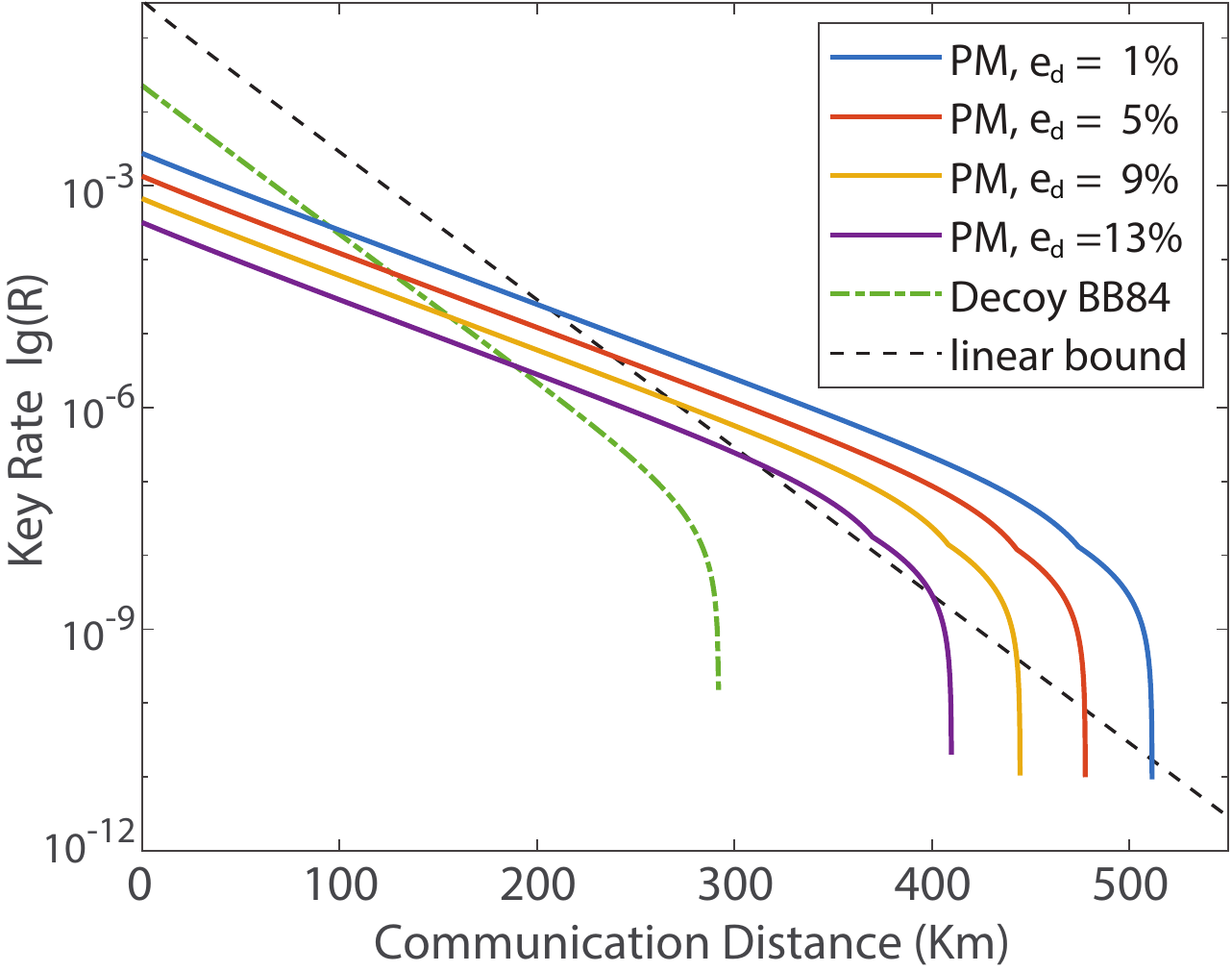}}
\caption{Rate-distance performance of PM-QKD under different system misalignment error rates $e_d$. For BB84, $e_d=1.5\%$. The non-smooth point indicates the places where the key contribution from unaligned groups with $j_s \neq 0$ turns to 0.} \label{fig:InfKey}
\end{figure}

\section{Finite size performance}
With the symmetry-based security proof of PM-QKD, the finite-size analysis can be much simplified. In a complete finite-size analysis, we should take the cost and failure probability of channel authentication, error verification, privacy amplification, and parameter estimation into account. However, the cost of the first three steps are negligible comparing to the one in parameter estimation. When the final key length is much larger than $37$ bits, we can ignore the corresponding failure probability with a constant secret key cost \cite{Fung2010Practical}. For simplicity, we ignore these parts in our analysis. The phase error estimation is the at the core of finite-size analysis. According to Eqs.~\eqref{eq:EX} and \eqref{eq:qk}, our task is to estimate the number of clicks caused by odd-photon fraction in the phase groups $J$, with signal intensity $\mu_a = \mu_b = \mu/2$. The Chernoff bound is applied to bound the statistical fluctuation of decoy parameters \cite{zhang2017improved}. We leave the details of the finite decoy-state analysis in Appendix \ref{Sc:finite}.

To demonstrate the practicality of PM-QKD, we perform simulation with finite data sizes, as shown in Fig.~\ref{fig:FiNed}. The key rate beat the linear bound at 270 km under the condition where data size $N=1\times 10^{12}$ and system misalignment error $e_d =3\%$. When the system misalignment error is $6\%$, which can be easily realizd in current experimental implementation, linear bound is exceeded at a similar length of 270 km, where, as an expense, the data size should be enlarged into $N=1\times 10^{13}$. Note that in the decoy analysis, the rounds with mismatched phases are also used for parameter estimation, which is substantiated by the fact that the single-photon state is $\rho^{1} = \frac{1}{2}(\ket{01}_{AB}\bra{01} + \ket{10}_{AB}\bra{10})$, regardless of the sending intensity and random phase difference. With this observation, the size of available data for parameter estimation is enlarged, which marginally reduces the impact of statistical fluctuation and results in a higher key rate.

\begin{table}[htbp]
\caption{List of parameters for the simulations shown in Fig.~\ref{fig:InfKey} and \ref{fig:FiNed}. The failure probability $\epsilon$ and sending rounds $N$ is used for the finite data size analysis in Fig.~\ref{fig:FiNed}.} \label{tab:para}
\begin{tabular}{cccccc}
  \hline
  Parameters & Values \\
  \hline
  Dark count rate $p_d$ & $1\times 10^{-8}$ \\
  Error correction efficiency $f$ & $1.1$  \\
  Detector efficiency $\eta_d$ & $20\%$  \\
  Number of phase slices $D$ & $16$  \\
  BB84 misalignment error $e_d^{(BB84)}$ & $1.5\%$ \\
  \hline
  Failure probability $\epsilon$ & $1.7\times 10^{-10}$ \\
  Sending rounds $N$ & $1\times 10^{12}$ or $1\times 10^{13}$ \\
  \hline
\end{tabular}
\end{table}

\begin{figure}[htbp]
\includegraphics[width=8cm]{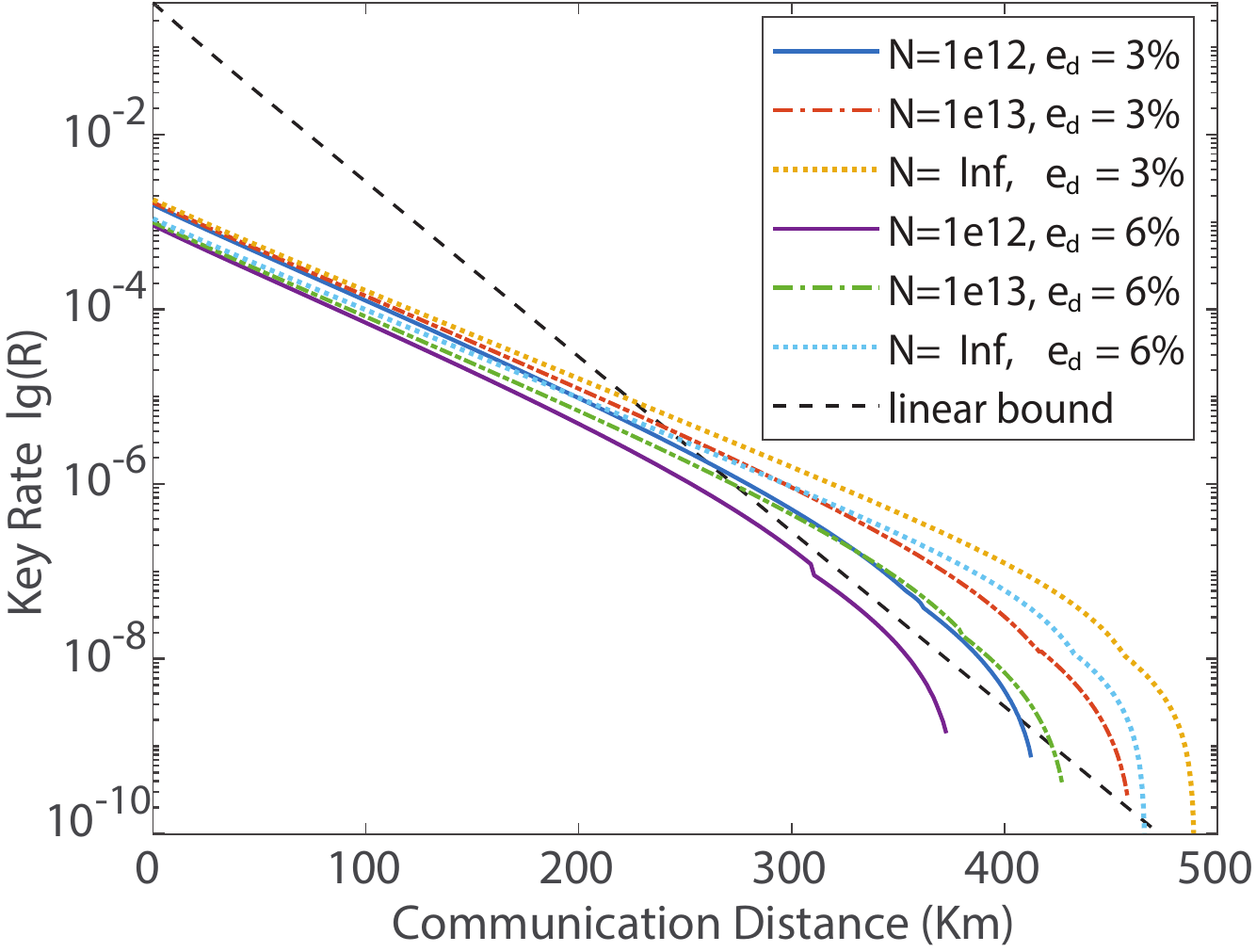}
\caption{Rate-distance performance of PM-QKD under the data size $N = 1\times 10^{12}, 1\times 10^{13}$ or infinitely large, and misaligned error $e_d = 3\%$ or $6\%$.} \label{fig:FiNed}
\end{figure}

\section{Concluding remarks}
We analyze the security of a generalized version of PM-QKD protocol, from which we provide a new perspective where the privacy of QKD originates from the encoding symmetry.
In PM-QKD, the parity symmetry protects the privacy of parity states. In the same manner, here we conjecture that the translation symmetry of encoded state in the round-robin differential-phase-shifting protocol may explain why the information-leakage will not be affected by the channel noise, and leave it for future works.

The symmetry-protected quality not only makes PM-QKD robust against channel disturbance, but also simplifies the parameter estimation and finite-size analysis considerably. With improved decoy-state analysis, we have demonstrated the capability of PM-QKD to surpass linear bound with data size $N=1\times 10^{13}$, currently accessible experimental devices, and high noise level of $6\%$.
Note that, the discrete phase randomization with $\{\phi = 2\pi \frac{j}{D}\}_{j=0}^{D-1}$ and imbalanced Alice and Bob's signal intensity $\mu_a, \mu_b$ will not destroy the parity-symmetry essentially. This implies a natural extension of PM-QKD analysis to the cases with few discrete random phases and imbalanced intensity arrangement of Alice and Bob.

Due to the universality of encoding symmetry and existence of symmetric states, we expect this symmetry-based analysis will benefit the security proof of a large variety of QKD protocols. For example, the analysis of  encoding operation with parity symmetry $U^2=I$ in this work can be extended to the $n$-fold rotational symmetry case, i.e., $U^n=I$, where $n\ge 2$. Moreover, this symmetry-based analysis can be extended to the case where the security proof is not obvious in an usual complementarity-based security view, for example, the continuous-variable QKD protocol.

\begin{acknowledgments}
We acknowledge M.~Koashi, N.~L\"utkenhaus, and H.~Zhou for the insightful discussions. This work was supported by the National Natural Science Foundation of China Grants No.~11875173 and No.~11674193, the National Key R\&D Program of China Grants No.~2017YFA0303900 and No.~2017YFA0304004, and the Zhongguancun Haihua Institute for Frontier Information Technology. P.~Z.~and W.~W.~contribute equally to this work.
\end{acknowledgments}

We provide detailed security analysis and simulation methods. In Section \ref{Sc:Proof}, we present the security proof for the symmetric-encoding QKD. In Section \ref{Sc:SecurePM}, we apply the symmetry-based security proof to the PM-QKD protocol, with practical issues taken into account. In Section \ref{Sc:finite}, we give the detailed finite data size analysis. Finally, we list the steps of numerical simulation in Section \ref{Sc:simulationformula}.

\begin{appendix}

\section{Proof of symmetry-protected privacy} \label{Sc:Proof}
In this section, we provide security analysis of symmetric-encoding QKD. We first review the security proof based on phase error correction \cite{Lo1999Unconditional,Shor2000Simple}. Then, we present a general entanglement-based symmetric-encoding protocol by establishing the link between symmetric states and perfect privacy.


Denote $\mathcal{D}(\mathcal{H}^A)$ as the space of density operators acting on $\mathcal{H}^A$ and $\mathcal{L}(\mathcal{H}^A)$ as the space of linear operators acting on $\mathcal{H}^A$. For a qubit system $A'$, the Hilbert space is denoted by $\mathcal{H^{A'}}$. The Pauli operators on $\mathcal{H^{A'}}$ are $X_{A'}, Y_{A'}$ and $Z_{A'}$. The eigenstates of $X_{A'}, Y_{A'}$ and $Z_{A'}$ are $\{\ket{\pm}_{A'}\}$, $\{\ket{\pm i}_{A'}\}$, and $\{\ket{0}_{A'},\ket{1}_{A'}\}$, respectively. The $X$-basis measurement is denoted by $M_X: \{\ket{+}_{A'}\bra{+}, \ket{-}_{A'}\bra{-} \}$. The $Z$-basis measurement is denoted by $M_Z: \{\ket{0}_{A'}\bra{0}, \ket{1}_{A'}\bra{1} \}$. The four Bell state are
\begin{equation}
\begin{aligned}
\ket{\Phi_\pm}&=\frac{1}{\sqrt{2}}(\ket{00}\pm\ket{11}),\\ 
\ket{\Psi_\pm}&=\frac{1}{\sqrt{2}}(\ket{01}\pm\ket{10}).\\
\end{aligned}
\end{equation}
For a qubit $A'$ and a qudit $A$, the controlled-$U$ operation is defined as
\begin{equation}
C_{A'A} = \ket{0}_{A'}\bra{0}\otimes I_A + \ket{1}_{A'}\bra{1}\otimes U_A.
\end{equation}

\subsection{Security definition}
In QKD, the two communication parties, Alice and Bob, will finally obtain a pair of bit strings $S_a$ and $S_b$ with a length of $N_k$ (if the protocol succeeds), which can be correlated to a quantum state held by Eve. The joint state $\rho_{ABE}$ is a classical-classical-quantum state,
\begin{equation} \label{eq:rho_ABE}
\rho_{ABE} = \sum_{s_a, s_b} Pr_{S_a, S_b}(s_a,s_b) \ket{s_a}_A\bra{s_a} \otimes \ket{s_b}_B\bra{s_b} \otimes \rho_E^{(s_a,s_b)},
\end{equation}
where $S_a, S_b$ are random variables and $s_a, s_b \in \{0,1\}^{N_k}$ are the values. In particular, an ideal key state held by Alice and Bob is described by the private state,
\begin{equation}
\rho_{ABE}^{ideal} = (2^{N_k})^{-1} \sum_{s} \ket{s}_A\bra{s} \otimes \ket{s}_B\bra{s} \otimes \rho_E,
\end{equation}
where $s_a=s_b=s$ implies that Alice and Bob hold the same string, and $\rho_E$ is independent of $s$, that is, Eve have no information on the key string variable $S$.

A QKD protocol is defined to be $\epsilon$-secure, if the final distilled state $\rho_{ABE}$ is $\epsilon$-closed to any private state $\rho_{ABE}^{ideal}$ with a proper chosen $\rho_E$
\begin{equation}
\min_{\rho_E} \dfrac{1}{2}|| \rho_{ABE} - \rho_{ABE}^{ideal} ||_1 \leq \epsilon,
\end{equation}
where $||A||_1 \equiv Tr[\sqrt{A^\dag A}]$ is the trace norm.

Usually, we would like to decompose the secret definition to two parts, secrecy and correctness. A QKD protocol is defined to be $\epsilon_{cor}$-correct, if the probability distribution $Pr_{S_a,S_b}(s_a, s_b)$ of the final state $\rho_{ABE}$ in Eq.~\eqref{eq:rho_ABE} satisfies
\begin{equation}
Pr_{S_a,S_b}(s_a \neq s_b) \leq \epsilon_{cor}.
\end{equation}
A QKD protocol is defined to be $\epsilon_{sec}$-secret, if the state $\rho_{AE}$ is closed to the single-party private state $\rho_{AE}^{ideal}$
\begin{equation}
\min_{\rho_E} \dfrac{1}{2}|| \rho_{AE} - \rho_{AE}^{ideal} ||_1 \leq \epsilon_{sec},
\end{equation}
where $\rho_{AE}^{ideal}\equiv (2^{N_k})^{-1} \sum_{s} \ket{s}_A\bra{s} \otimes \rho_E$. If a QKD protocol is $\epsilon_{cor}$-correct and $\epsilon_{sec}$-secret, then it is $(\epsilon_{cor} + \epsilon_{sec})$-secure \cite{koashi2009simple}.

\subsection{Security proof based on phase error correction} \label{Sc:Koashi}
Here, we briefly review the main idea of phase-error-based security proof, which is first proposed by Lo-Chau \cite{Lo1999Unconditional}, reduced to prepare-and-measure scheme later by Shor-Preskill \cite{Shor2000Simple}, and improved by Koashi later \cite{koashi2009simple}.

In an entanglement-based QKD protocol, Alice and Bob will finally share a large bipartite state $\rho_{AB}$. Denote Alice's and Bob's subspace in a single run as $\mathcal{H}_A$ and $\mathcal{H}_B$. Here, the encrypted error correction is used to decouple the error correction and privacy amplification. The protocol is presented as below.

\textbf{\uline{Actual protocol}}
\begin{enumerate}
\item
(State distribution) Alice and Bob share a bipartite state $\rho_{AB}$ in the space $(\mathcal{H}_A\otimes\mathcal{H}_B)^{\otimes n}$, where $n$ is the total number of runs.
\item
(Measurement) Alice and Bob measure their systems $\mathcal{H}_A^{\otimes n}$ and $\mathcal{H}_B^{\otimes n}$, respectively. Suppose the measurement results can be described by $n$-bit strings $\kappa_A$ and $\kappa_B$.
\item
(Error correction) They reconcile the key strings through an encrypted classical channel consuming $l_{ec}$-bits of secret key. They agree on an $n$-bit raw key string $\kappa_{rec}$ except for a small failure probability $\epsilon_{ec}$.
\item
(Privacy amplification) Alice randomly chooses $n-m$ strings $\{V_k\}_{k=1,...,n-m}$ of $n$-bit, which are linearly independent, and announces the strings to Bob. The final key length is $n-m$, where the $k$-th key bit is $\kappa_{rec}\cdot V_k$. Denote the final key as $\kappa_{fin}$.
\end{enumerate}

Note that the error correction can be easily conducted just by classical information theory. The only remaining concern is to quantize the information leakage $m$ on $\kappa_{rec}$. Note that, Eve's knowledge on $\kappa_{rec}$ will not change after any operation on $\mathcal{H}_A$ and $\mathcal{H}_B$ that keeps Eve's state and the raw key bits $\kappa_{rec}$. Based on this argument, we construct the following virtual protocol.

\textbf{\uline{Virtual Protocol}}
\begin{enumerate}
\item
(State distribution) Alice and Bob share a bipartite state $\rho_{AB}\in(\mathcal{H}_A\otimes\mathcal{H}_B)^{\otimes n}$.
\item
(Squashing) They apply operation $\Lambda$ on $\rho_{AB}$ and convert it to a key space $\mathcal{K}^{\otimes n}$ and ancillary space $\mathcal{H}_R$, i.e. $\Lambda(\rho_{AB})\in\mathcal{K}^{\otimes n}\otimes\mathcal{H}_R$.
\item
(Measurement) They measure $\mathcal{K}^{\otimes n}$ on the $Z$ basis and obtain $\kappa_{rec}$. They then measure $\mathcal{H}_R$ by $M_R$. Suppose the measurement result of $M_R$ is $\gamma$.
\item
(Privacy amplification) They randomly chooses $n-m$ linearly independent $n$-bit strings $\{V_k\}_{k=1,...,n-m}$ and announce them. The final key length is $n-m$, and the $k$-th key bit is $\kappa_{rec}\cdot V_k$. Denote the final key as $\kappa_{fin}$.
\end{enumerate}

The operation $\Lambda$ and measurement $M_R$ can be chosen freely, which are only subjected to the requirement that the $Z$-basis measurement statistics on $\mathcal{K}^{\otimes n}$ is the same as $\kappa_{rec}$ in actual protocol. Therefore, Eve's knowledge on $\kappa_{rec}$ in the virtual protocol is the same as in the actual protocol. The secrecy of $\kappa_{rec}$ in virtual protocol is the same as the one in practical protocol.

The core observation in the security proof based on phase error correction is that, with a proper choice of $\Lambda$ and $M_R$, the security of $Z$-basis measurement result $\kappa_{rec}$ can be reflected on the predictability of $X$-basis measurement result $T_\gamma$, given the measurement outcome $\gamma$ on $\mathcal{H}_R$. Denote $|T_\gamma|$ as the size of possible $X$-basis measurement outcomes.

\begin{lemma} \label{lem:Koashi} (Koashi'09 \cite{koashi2009simple}) If the chosen $\Lambda$ and $M_R$ in the above virtual protocol meets the requirements:
\begin{enumerate}
	\item The $Z$-basis measurement statistics on $\mathcal{K}^{\otimes n}$ is the same as $\kappa_{rec}$ in actual protocol,
	\item Given each measurement outcome $\gamma$ on $\mathcal{H}_R$, the size of $X$-basis measurement outcome on $\mathcal{K}^{\otimes n}$ is bounded by $|T_\gamma| \leq 2^{n\xi}$, except for a small probability $\epsilon_{T}$,
\end{enumerate}
then the virtual protocol is $\sqrt{\epsilon_{T}'}$-secret and $\epsilon_{ec}$-correct, thus $(\sqrt{\epsilon_{T}'}+\epsilon_{ec})$-secure, where $\epsilon_{T}'=\epsilon_{T} + 2^{-n\zeta}$ and $m=n(\xi + \zeta)$.
\end{lemma}



Here we collect all the related small failure probabilities,
\begin{itemize}
	\item $\epsilon_{ec}$ is the failure probability of error correction, with affect the correctness of the protocol, which is determined by the method used in the error correction step;
	\item $\epsilon_{T}$ is the failure probability of bounding $|T_\gamma|$. Usually this is the failure probability $\epsilon_{pe}$ of estimating phase error number $n^{EX}$. This amount is determined by the method used to estimate the phase error;
	\item $\zeta$ is the extra amount of privacy amplification to enhance the privacy of the protocol, which can be determined arbitrarily according to the need for privacy. Usually we denote $\epsilon_{pa} = 2^{-n\zeta}$ as the failure probability of privacy amplification.
\end{itemize}

Usually, we introduce phase error number $n^{EX}$ to characterize the size of $|T_\gamma|$. Suppose that, in an ideal case, the $X$-basis measurement outcome $\mathcal{K}^{\otimes n}$ for a given $\gamma$ is deterministic $T^{(0)}_\gamma$. Then we can calculate the number of phase error $n^{EX}$ of given measurement result $T_\gamma$ from $T^{(0)}_\gamma$, defined as $n^{EX}(T_\gamma) \equiv wt(T_\gamma \oplus T^{(0)}_\gamma)$, where $wt(A)$ is the weight, i.e., number of non-zero elements in string $A$, and $\oplus$ is the modulo-$2$ addition. Denote $n^{EX}$ as the maximum value of $n^{EX}(T_\gamma)$ for a given set $\{T_\gamma\}$. Then the size of $\{T_\gamma\}$ is bounded by
\begin{equation}
|T_\gamma| = \sum_{k=0}^{n^{EX}} \binom{n}{k} < \binom{n}{n^{EX}+1} < 2^{nH(\frac{n^{EX}+1}{n})},
\end{equation}
where the first inequality holds when $(n^{EX}+1<n/3)$ \cite{Fung2010Practical}. The proof of the second inequality can be found in Ref.~\cite{cover2012elements}. From phase error correction point of view, this is essentially the typical space argument used in Shannon's theory.

Denote the average phase error number as $\bar{n}^{EX}$. The phase error rate is define to be $E^X \equiv \bar{n}^{EX}/n$. Note that the final key length is $n-m$, where $m = n\xi \geq nH(\frac{n^{EX}+1}{n})=nH(E^X)$, we conclude that an approximate proportion $H(E^X)$ of raw key bits is sacrificed in privacy amplification. Taken the reconciliation cose $l_{ec}$ into account, the net key generation length is
\begin{equation}
K = n - m - l_{ec} \geq n(1 - H(E^X)) - l_{ec}.
\end{equation}

One should note that, in Koashi's security proof, there is \emph{a degree of freedom on choosing the definition of phase error $E^{X}$}, based on the operation $\Lambda$ and the ancillary measurement $M_R$, as long as Lemma \ref{lem:Koashi} holds. This endows a large flexibility when we analyze the security of QKD.

\subsection{Perfect privacy induced by encoding symmetry} \label{Sc:PM parity}
We prove the symmetric-encoding QKD based on the aforementioned security proof, which essentially generalizes the one employed in the security proof of PM-QKD \cite{ma2018phase}. First, we introduce a general entanglement-based symmetric-encoding QKD protocol, as shown in Fig.~\ref{fig:Pro1}. Here, Alice and Bob share a two-qudit state $\sigma_{AB}$ on system $A$ and $B$, at the beginning of each run. Alice holds an extra qubit $A'$ on state $\ket{+i}$ initially, and she creates entanglement between $A$ and $A'$ by applying the following control-$U$ operation
\begin{equation}
C_{A'A}(U) = \ket{0}_{A'}\bra{0}\otimes I_A + \ket{1}_{A'}\bra{1}\otimes U_A.
\end{equation}
She then sends qudit $A$ out to Eve. Similarly, Bob holds $B'$, performs $C_{B'B}(U)$ on $B$ and $B'$, and sends $B$ to Eve. Eve performs joint measurement on system $A$ and $B$ to create entanglement on qubit system $A'$ and $B'$. Here, the encoding unitary $U_A$ in $C_{A'A}(U)$ meets the requirement of symmetric encoding,
\begin{equation}
U^2 = U^\dag U = I.
\end{equation}

\begin{figure}[htbp]
\includegraphics[width=8cm]{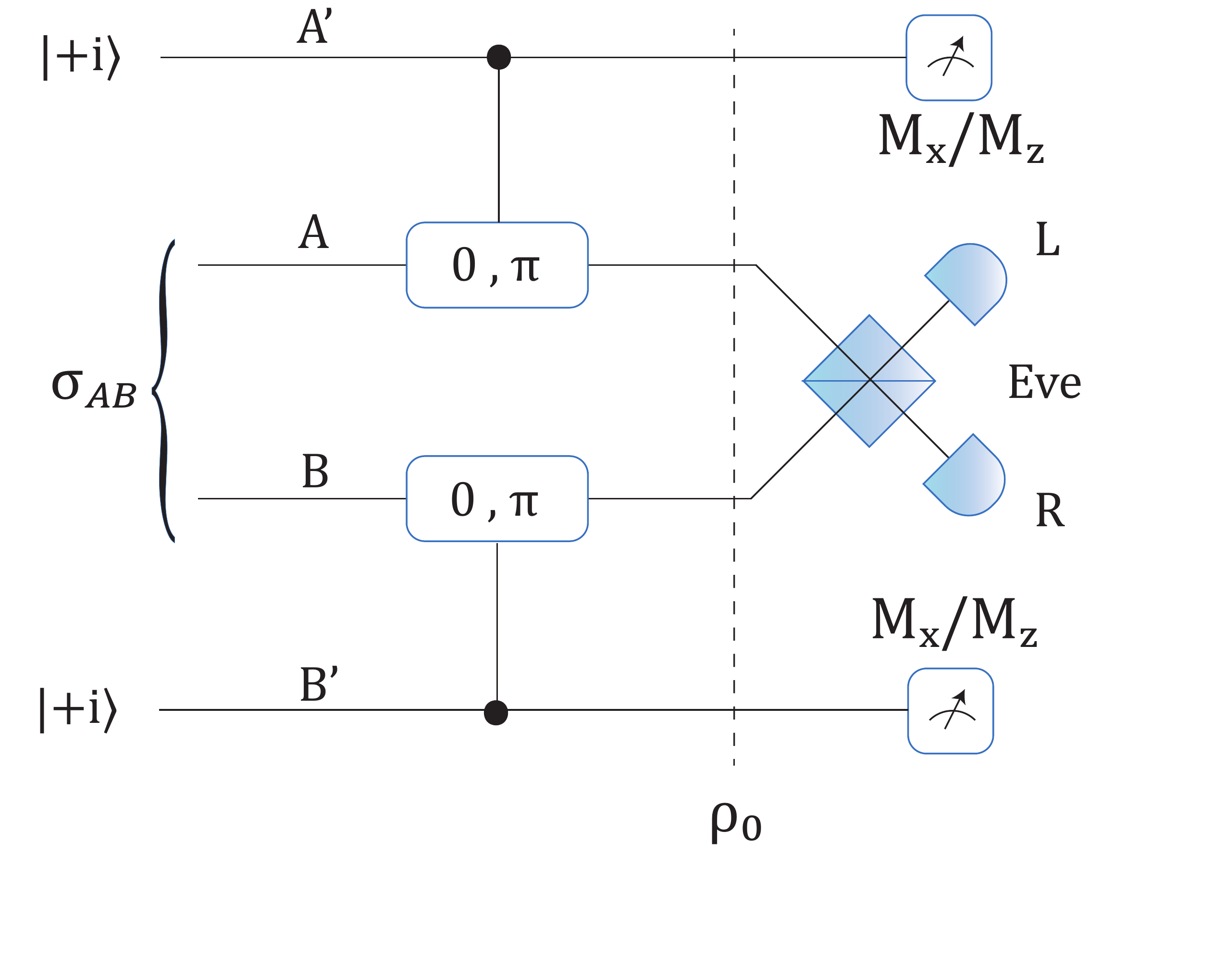}
\caption{Schematic diagram for a general entanglement-based version of symmetric-encoding protocol. When $\sigma_{AB}$ is a parity state, then protocol is perfectly private without any information leakage.} \label{fig:Pro1}
\end{figure}

In that case, the eigenvalue of $U_A$ will be either $1$ or $-1$. Denote the subspaces spanned by the eigenvectors with eigenvalue $1$ and $-1$ to be the \uline{even space} $\mathcal{H}^{even}_A$ and the \uline{odd space} $\mathcal{H}^{odd}_A$. For a joint unitary $U_A \otimes U_B$, the eigenvalue will also be $1$ or $-1$. The even space and the odd space of the joint $A,B$ Hilbert space is
\begin{equation}
\begin{aligned}
\mathcal{H}^{odd}_{AB} &= (\mathcal{H}^{odd}_{A}\otimes \mathcal{H}^{even}_{B})\oplus (\mathcal{H}^{even}_{A}\otimes \mathcal{H}^{odd}_{B}), \\
\mathcal{H}^{even}_{AB} &= (\mathcal{H}^{odd}_{A}\otimes \mathcal{H}^{odd}_{B})\oplus (\mathcal{H}^{even}_{A}\otimes \mathcal{H}^{even}_{B}). \\
\end{aligned}
\end{equation}

\begin{definition} \label{def:oddevenAB}
A state $\rho$ on two qudits $A,B$ is an \uline{odd state} with respect to $U_A\otimes U_B$ iff $\rho\in \mathcal{H}^{odd}_{AB}$ and an \uline{even state} with respect to $U_A\otimes U_B$ iff $\rho\in \mathcal{H}^{even}_{AB}$. Both odd states and even states are called \uline{parity states}.
\end{definition}

\begin{corollary} \label{prop:mixparity}
A state $\rho$ keeps invariant under the transformation of $U_A\otimes U_B$, i.e., $\rho = U_A\otimes U_B \rho U_A^\dag\otimes U_B^\dag$, iff $\rho$ is a mixture of parity states.
\end{corollary}

\begin{proof}
Obviously, the mixture of parity states satisfies the encoding symmtry $\rho = U_A\otimes U_B \rho U_A^\dag\otimes U_B^\dag$. Here we prove all the states $\rho$ with encoding symmtry are mixture of parity states.

Since $(U_A\otimes U_B)^2 = I$, the eigenvalue of it should be either $-1$ or $1$. Denote the eigenbasis of $U_A\otimes U_B$ as $\{\ket{p,i}\}$ with eigenvalues $\{(-1)^p\}$, where $p=0,1$ denotes the even subspace $\mathcal{H}^{even}_{AB}$ or odd subspace $\mathcal{H}^{odd}_{AB}$, while $i$ denotes the inner degeneracy in the even or odd subspace.

We expand $\rho$ in the eigenbasis of $U_A\otimes U_B$,
\begin{equation}
\rho = \sum_{p,q,i,j} c_{p,q,i,j} \ket{p,i}\bra{q,j}.
\end{equation}

Therefore,
\begin{equation}
U_A\otimes U_B \rho U_A^\dag \otimes U_B^\dag = \sum_{p,q,i,j} (-1)^{p-q} c_{p,q,i,j}  \ket{p,i}\bra{q,j},
\end{equation}
and $\rho = U_A\otimes U_B \rho U_A^\dag \otimes U_B^\dag$ implies that
\begin{equation}
c_{p,q,i,j} = 0 \text{ if } p\ne q,
\end{equation}
that is, the off diagonal space between odd and even states is a null space.

Moreover, due to the degeneracy of odd and even subspace, it is alway possible to find an eigenbasis of $U_A\otimes U_B$ as $\{\ket{p,i}\}$ to diagonalize the parity states.
\end{proof}

We will first consider the following entanglement-based QKD protocol, and analyze how the symmetry of the state will preserve QKD privacy.

\textbf{\uline{Protocol I}}
\begin{enumerate}
\item \label{StepPrep}
State preparation: Alice and Bob share a known state $\sigma_{A,B}$ on two qudits $A,B$, at the beginning of each run. They initialize their qubits $A', B'$ in $\ket{+i}$. Alice applies the control gate $C_{A'A}(U)$ to $A'$ and $A$. Similarly, Bob applies $C_{B'B}(U)$ to $B'$ and $B$.


\item \label{StepMeasure}
Measurement: Alice and Bob send the two optical pulses $A$ and $B$ to an untrusted party, Eve, who is supposed to perform joint measurement and announce the detection results. The ideal measurement is to discriminate whether the state is $\sigma_{AB}$ or $(U_A\otimes I) \sigma_{AB} (U_A\otimes I)^\dag$.

\item \label{StepAnnounce}
Announcement: Eve announces the detection result for each round. Based on Eve's announcement, Bob decides whether or not to apply $Y$ operation on his qubit $B'$.

\item \label{StepSifting}
Sifting: Given a specific announcement of Eve, Alice and Bob keep the qubits of systems $A'$ and $B'$.

Alice and Bob perform the above steps for many rounds and end up with a joint $2n$-qubit state $\rho_{A' B'}\in(\mathcal{H}_{A'}\otimes\mathcal{H}_{B'})^{\otimes n}$. 


\item\label{StepClassical}
Key generation: Alice and Bob perform local $Z$ measurements on $\rho_{A' B'}$ to generate raw data string $\kappa_A$ and $\kappa_B$. They reconcile the key string to $\kappa_{rec}$ by an encrypted classical channel, with the consummation of $l_{ec}$-bit keys. Here we set Alice's key as the final reconciled key $\kappa_A = \kappa_{rec}$.
\end{enumerate}


Denote the whole $2n$ state Alice and Bob share after the quantum step as $\rho_{A' B'}$, and the partial-traced state of the $m$-th round as $\rho_{A' B'}^{(m)}$.

\begin{lemma} \label{Lem:parity}
In Protocol I, if the optical state $\sigma$ Alice and Bob share during the $m$-th run is an odd state, then the two-qubit state Alice and Bob finally obtain, $\rho_{A' B'}^{(m)}$, is in the subspace spanned by $\{\Phi^+, \Psi^-\}$; and if $\sigma$ is an even state, then $\rho_{A' B'}^{(m)}$ is in the subspace spanned by $\{\Phi^-, \Psi^+\}$. Here $m\in\{1,2,\dots,n\}$.
\end{lemma}

\begin{proof}

First consider the pure state $\sigma = \ket{\psi}\bra{\psi}$. The joint state on system $A',B',A,B$ before the $C(\pi)$ operations is
\begin{equation}
\ket{+i+i}_{A'B'}\ket{\psi}_{AB} = \dfrac{1}{2} [(\ket{00} - \ket{11}) + i (\ket{01} + \ket{10}) ]_{A'B'} \ket{\psi}_{AB}.
\end{equation}
Note that for odd state $\ket{\psi_o}_{AB}$,
\begin{equation}
\begin{aligned}
U_{A} \otimes U_{B} \ket{\psi_o}_{AB} &=  - \ket{\psi_o}_{AB}, \\
U_{A} \otimes I_B \ket{\psi_o}_{AB} &=  - I_A \otimes U_{B} \ket{\psi_o}_{AB},
\end{aligned}
\end{equation}
and for even state $\ket{\psi_e}_{AB}$,
\begin{equation}
\begin{aligned}
U_{A} \otimes U_{B} \ket{\psi_e}_{AB} &= \ket{\psi_e}_{AB}, \\
U_{A} \otimes I_B \ket{\psi_e}_{AB} &= I_A \otimes U_{B}\ket{\psi_e}_{AB},
\end{aligned}
\end{equation}
Therefore, for odd pure state $\ket{\psi_o}$ input, the state after the $C_{A'A}(U)\otimes C_{B'B}(U)$ operations is
\begin{widetext}
\begin{equation}
\begin{aligned}
\ket{\Psi_o} &= C_{A'A}(U) \otimes C_{B'B}(U) (\dfrac{1}{2} [(\ket{00} - \ket{11}) + i (\ket{01} + \ket{10}) ]_{A'B'} \ket{\psi_o}_{AB}) \\
= & \dfrac{1}{2} \left\{ [\ket{00}_{A'B'} - \ket{11}_{A'B'}(U_{A} \otimes U_{B})] + i [\ket{01}_{A'B'}( U_{A} \otimes I_B ) + \ket{10}_{A'B'}( I_A \otimes U_{B}(\pi) ] \right\} \ket{\psi_o}_{AB} \\
= & \dfrac{1}{2} [(\ket{00} + \ket{11})_{A'B'} \ket{\psi_o}_{AB} + i (\ket{01} - \ket{10})_{A'B'} ( U_{A} \otimes I_B )\ket{\psi_o}_{AB} ],
\end{aligned}
\end{equation}
\end{widetext}
and for even pure state $\ket{\psi_e}$ input, the state after the $C_{A'A}(U)\otimes C_{B'B}(U)$ operations is
\begin{widetext}
\begin{equation}
\begin{aligned}
\ket{\Psi_e} &= C_{A'A}(U) \otimes C_{B'B}(U) (\dfrac{1}{2} [(\ket{00} - \ket{11}) + i (\ket{01} + \ket{10}) ]_{A'B'} \ket{\psi_e}_{AB}) \\
= & \dfrac{1}{2} \left\{ [\ket{00}_{A'B'} - \ket{11}_{A'B'}(U_{A} \otimes U_{B})] + i [\ket{01}_{A'B'}( U_{A} \otimes I_B ) + \ket{10}_{A'B'}( I_A \otimes U_{B}(\pi) ] \right\} \ket{\psi_e}_{AB} \\
= & \dfrac{1}{2} [(\ket{00} - \ket{11})_{A'B'} \ket{\psi_e}_{AB} + i (\ket{01} + \ket{10})_{A'B'} ( U_{A} \otimes I_B )\ket{\psi_e}_{AB} ],
\end{aligned}
\end{equation}
\end{widetext}

For the odd state case, if we partially trace out system $A,B$ in $\ket{\Psi_o}$, the state $\rho_{A',B'}^{(m)}$ is in the subspace of $\{\Phi^+, \Psi^-\}$. Whatever Eve's announcement afterward is, the possible operation on $\rho_{A',B'}^{(m)}$ is either $I_{A'}\otimes I_{B'}$ or $I_{A'}\otimes Y_{B'}$. Note that $(I_{A'}\otimes Y_{B'})\ket{\Phi^+} = \ket{\Psi^-}$, hence the state $\rho_{A',B'}^{(m)}$ is still in the subspace of $\{\Phi^+, \Psi^-\}$. Similarly, for the even state case, the state $\rho_{A',B'}^{(m)}$ is in the subspace of $\{\Phi^-, \Psi^+\}$.

For general mixed states, we can regard them as mixtures of pure parity states,
\begin{equation}
\begin{aligned}
\sigma_{odd} &= \sum_{i} p_i \ket{\psi_{o}^{(i)}}\bra{\psi_{o}^{(i)}}, \\
\sigma_{even} &= \sum_{i} p_i \ket{\psi_{e}^{(i)}}\bra{\psi_{e}^{(i)}}.
\end{aligned}
\end{equation}
This is equivalent to Charlie sending out $\ket{\psi_{o(e)}^{(i)}}$ with probability $p_i$. For each odd pure state component, the left qubit state $\rho_{A',B'}^{(m)}$ is in the subspace of $\{\Phi^+, \Psi^-\}$, therefore their mixture are still in this subspace. Similar arguments hold for the even parity states.
\end{proof}

If the state $\rho_{A',B'}^{(m)}$ is in the subspace of $\{\Phi^+, \Psi^-\}$, then the measurement value of $X\otimes X$ and $Z\otimes Z$ are always the same; on the contrary, if $\rho_{A',B'}^{(m)}$ is in the subspace of $\{\Phi^-, \Psi^+\}$, then the measurement value of $X\otimes X$ and $Z\otimes Z$ are always different. From Shor-Preskill's view, when we detect the $Z$-basis error location, we can apply $I_{A'}\otimes Y_{B'}$ operation to correct the $X$-basis error and $Z$-basis error simultaneously, thus the error correction can be taken only once.

Following the method in Ref.~\cite{koashi2009simple}, we construct a virtual protocol, named Protocol Ia, and there is a freedom to choose the definition of phase error. The information leakage can be characterized by the phase error with an optimized definition.

\textbf{\uline{Protocol Ia}}
\begin{enumerate}
\item With the same manner, Alice and Bob perform the steps \ref{StepPrep}$\sim$\ref{StepSifting} in Protocol I for many rounds and end up with a joint $2n$-qubit state $\rho_{A'B'}$.


\item 
Key generation: Alice and Bob first perform a control-$(-Z)$ gate on each round of state $\{\rho_{A' B'}^{(m)}\}_{m=1}^{n}$. After that, Bob perform local $X$ measurements on system $B'$ with measurement outcome $\gamma_{B'}$, and Alice perform local $Z$ measurements to obtain the final key $\kappa_{rec}$.
\end{enumerate}

Denote the control-$(-Z)$ gate as $C(-Z)_{A',B'}$, then
\begin{equation}
C(-Z)_{A',B'} = \ket{0}_{A'}\bra{0} \otimes I + \ket{1}_{A'} \bra{1} \otimes (-Z)
\end{equation}

Note that
\begin{equation}
\begin{aligned}
C(-Z)_{A',B'} \ket{\Phi^+} = \ket{\Phi^+}, \quad C(-Z)_{A',B'} \ket{\Phi^-} = \ket{\Phi^-}, \\
C(-Z)_{A',B'} \ket{\Psi^+} = \ket{\Psi^-}, \quad C(-Z)_{A',B'} \ket{\Psi^-} = \ket{\Psi^+}.
\end{aligned}
\end{equation}

Since control-$(-Z)$ gate commutes with $Z$-measurement on system $A'$, applying control-$(-Z)$ gate will note affect the $Z$-measurement outcome. Therefore, Alice's $Z$-measurement results $\kappa_{rec}$ should be the same as the one in Protocol I, and the information leakage of $\kappa_{rec}$ to Eve is the same.

If $\sigma$ is an odd state, then $\rho^{(m)}_{A',B'}$ is in the subspace spanned by $\{\Phi^+, \Psi^-\}$. After Alice and Bob apply $C(-Z)_{A',B'}$ gate, $\rho^{(m)}_{A',B'}$ is in the subspace spanned by $\{\Phi^+, \Psi^+\}$. Note that
\begin{equation}
(X_{A'} \otimes X_{B'}) \ket{\Phi^+} = (X_{A'} \otimes X_{B'}) \ket{\Psi^+} = 1,
\end{equation}
therefore, after the $C(-Z)_{A',B'}$ gate, the $X$-basis measurement result on system $A'$ and $B'$ should always be the same. If we set the outcome $\gamma_{B'}$ as the correct $X$-basis measurement result of $B'$, then the phase error number is $0$.

Similarly, if $\sigma$ is a even state, then $\rho^{(m)}_{A',B'}$ is in the subspace spanned by $\{\Phi^-, \Psi^+\}$. After Alice and Bob apply $C(-Z)_{A',B'}$ gate, $\rho^{(m)}_{A',B'}$ is in the subspace spanned by $\{\Phi^+, \Psi^+\}$. Note that
\begin{equation}
(X_{A'} \otimes X_{B'}) \ket{\Phi^-} = (X_{A'} \otimes X_{B'}) \ket{\Psi^-} = -1,
\end{equation}
therefore, after the $C(-Z)_{A',B'}$ gate, the $X$-basis measurement results on system $A'$ and $B'$ should always be different. If we set $-\gamma_{B'}$ as the correct $X$-basis measurement result of $A'$, then the phase error number is $0$.

\begin{definition}
In Protocol Ia, define the observed value $\frac{1}{2}(1 - \left<X_{A'} \otimes X_{B'}\right>)$ to be the \emph{odd phase error rate} $E^{X}_o$, and $\frac{1}{2}(1 + \left<X_{A'} \otimes X_{B'}\right>)$ to be the \emph{even phase error rate} $E^{X}_e$.
\end{definition}

Note that $E^{X}_e = 1 - E^{X}_o$. We have the following observation,

\begin{observation} \label{ob:0error}
In Protocol Ia, for an odd(even) state input $\sigma$, the odd(even) phase error rate is always $0$. That is, if Alice measures system $A'$ on $X$-basis instead of $Z$-basis, then the measurement results is certain, given Bob's measurement outcome $\gamma_{B'}$ on system $B'$.
\end{observation}

In the phase-error-based proof, we may choose the definition of phase error $E^X$ as long as the requirements in Lemma \ref{lem:Koashi} hold. Therefore, the odd states and even states $\sigma$ are both be the perfect source for Protocol I, with perfect privacy.

\begin{figure*}[htbp]
\includegraphics[width=12cm]{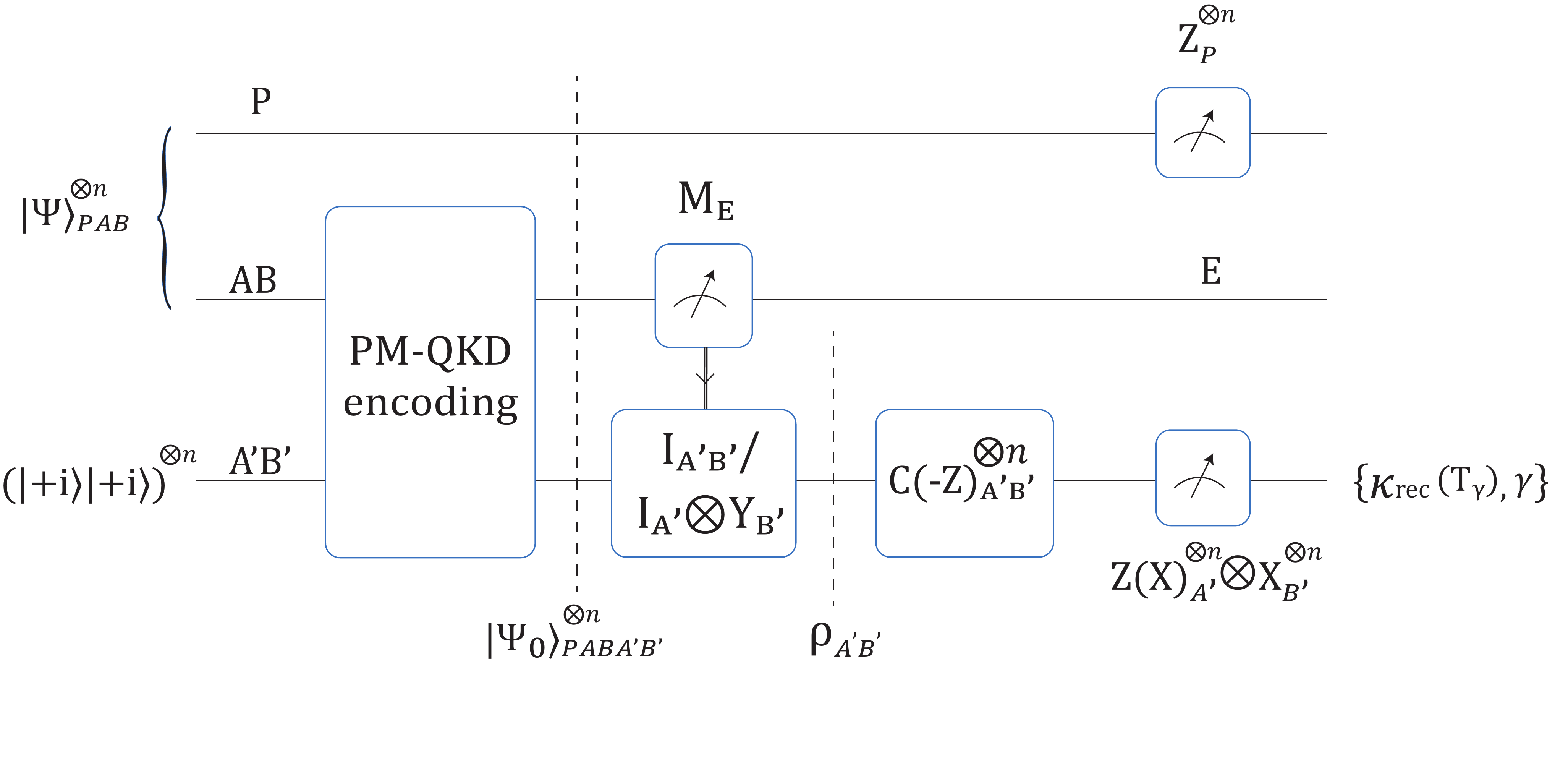}
\caption{Schematic diagram for the whole Protocol Ia/Ib when the input state is a general mixed parity state $\sigma_{AB}$. Here we post-select the rounds when Eve announces effective clicks. Alice and Bob prepare $n$-rounds of state $\ket{\Psi}_{PAB}$ in Eq.~\eqref{eq:PAB}. They performs the PM-QKD encoding by applying $C(U)$ gates on system $A',A$ and $B',B$, as is illustrated in Fig.~\ref{fig:Pro1}. After that, they send system $A,B$ to Eve, and performs $I_{A'B'}$ or $I_{A'}\otimes Y_{B'}$ operation on each round of system $A',B'$. They apply $C(-Z)$ gates on $A',B'$, and measure them by $Z_{A'}\otimes X_{B'}$ to get strings $\kappa_{rec}, \gamma$. In Protocol Ib, after Eve's announcement, they send system $P$ to Eve. In phase-error-based proof, Eve's knowledge on $\kappa_{rec}$ can be characterized by Alice's knowledge on her $X$-basis measurement result $T_\gamma$ on $A'$.} \label{fig:PAB}
\end{figure*}

The case will become nontrivial when Alice and Bob hold a mixture of odd and even states. The overall phase error rate cannot be zero, regardless of whether odd or even phase error definition is applied.

We write the state $\sigma_{AB}$ shared by Alice and Bob as
\begin{equation}
\sigma_{AB} = \sum_{i} p_o^{(i)} \ket{\psi_o^{(i)}}_{AB}\bra{\psi_o^{(i)}} + \sum_{j} p_e^{(j)} \ket{\psi_e^{(j)}}_{AB}\bra{\psi_e^{(j)}},
\end{equation}
with $\sum_{i}p_o^{(i)} + \sum_{j}p_e^{(j)} =1$. Introduce the purification of $\sigma$,
\begin{widetext}
\begin{equation} \label{eq:PAB}
\ket{\Psi}_{PAB} = \sum_{i} \sqrt{p_o^{(i)}} \ket{o,i}_{P}\ket{\psi_o^{(i)}}_{AB} + \sum_{j} \sqrt{p_e^{(j)}} \ket{e,j}_{P}\ket{\psi_e^{(j)}}_{AB},
\end{equation}
\end{widetext}
here system $P$ is a register to store the parity information. Suppose $\{\ket{p,i}_P\}$ is an orthogonal basis such that $\braket{p,i|q,j} = \delta_{pq} \delta_{ij}$, with $p=o, e$ denotes the parity of the according state, and $i$ denotes the index of odd/even pure states. Such basis is defined as $Z$-basis of system $P$.

Here we consider two cases, where Eve's knowledge on $\sigma_{AB}$ are different. First, we consider the case when Eve's state is decoupled from $\sigma_{AB}$ at the beginning, that is, $\rho_{ABE} = \sigma_{AB}\otimes \rho_E$. That is to say, Eve doesn't hold the purified system $P$, and we can virtually imagine that system $P$ is hold by Alice and Bob. In Fig.~\ref{fig:PAB} we draw the whole Protocol Ia with the ancillary $P$ presents. In Koashi's proof, when we characterize Eve's knowledge on Alice's final $Z$-basis measurement results on $A'$ $\kappa_{rec}$, we transform the problem to \emph{how Alice can predict the $X$-basis measurement result $\gamma$ in the presence of system $B'$ and $P$}. In this case, Alice can learn the parity information of each single round from the measurement result of $P$, where Alice can flip her X-basis measurement result on $\mathcal{K}$ to match $\gamma$ if the parity is not accordance with the definition choice of phase error. Therefore, by Observation \ref{ob:0error} above, Alice can perfectly predict the $X$-basis measurement result. Applying Lemma \ref{lem:Koashi}, we prove the perfect privacy of parity states.

\begin{theorem} \label{Ob:parity0}
In Protocol I and Ia, if Alice and Bob share a mixture of odd and even state $\sigma_{AB}$ at the beginning of each run, and Eve's state $\rho_E$ is decoupled from $\sigma_{AB}$, i.e., $\rho_{ABE} = \sigma_{AB} \otimes \rho_E$, then Eve has no information on Alice's $Z$-basis measurement result $\kappa_{rec}$, i.e. perfect privacy.
\end{theorem}

Now we consider the case when Eve holds the purification of $\sigma_{AB}$. We consider the following protocol.

\textbf{\uline{Protocol Ib}}
\begin{enumerate}
\item With the same manner, Alice and Bob perform the steps \ref{StepPrep}$\sim$\ref{StepSifting} in Protocol I for many rounds and end up with a joint $2n$-qubit state $\rho_{A'B'}$. \emph{After that, Alice and Bob send the purified system $P$ of each run to Eve}.


\item 
Key generation: Alice and Bob first perform a control-$(-Z)$ gate on each round of state $\{\rho_{A' B'}^{(m)}\}_{m=1}^{n}$. After that, Bob perform local $X$ measurements on system $B'$ with measurement outcome $\gamma_{B'}$, and Alice perform local $Z$ measurements to obtain the final key $\kappa_{rec}$.
\end{enumerate}

In Protocol Ib, a weaker version of Protocol I, Alice and Bob no longer holds the system $P$, which makes them unable to discriminate the parity information, and their ability to predict $T_\gamma$ is weakened. However, an important observation is that, due to system $P$ is sent to Eve after her announcement, the announcement result must be independent of system $P$, in which case, for a fixed Eve's measurement strategy, the state $\rho_{A'B'}$ in Protocol Ia and Ib will be the same.

In another aspect, without the assistance of parity information in $P$, Alice and Bob cannot deal with the odd and even rounds separately, and there is no longer perfect privacy. Alice's knowledge on $T_\gamma$ given $\gamma$ is characterized by the overall odd or even phase error. For the clicked $n$ rounds, if we measure the system $P$ on $Z$-basis, and there are $n_{even}$ rounds with even parity measurement results. Then the overall odd phase error rate is
\begin{equation} \label{eq:EXo}
E^X_o = \dfrac{n_{even}}{n}.
\end{equation}

From now on, We define the phase error $E^{ph}$ as the odd phase error rate $E^X_o$. As an expense to realize parity states by independent phase randomization with phase announcement, we cannot discriminate even and odd components any more, and should use a unified phase error definition, while Eve's announcement strategy is unchanged, and the local state $\rho_{A'B'}$ also keeps unchanged, which indicates that Protocol Ia and Protocol Ib are the same except in parity discrimination.

\section{Security proof of PM-QKD} \label{Sc:SecurePM}
Here we present security proof of PM-QKD, by reducing it to the symmetric-encoding QKD mentioned above. Especially, we use the decoy method to monitor the phase error number $n^{EX}$ in real experimental setting to determine how much proportion of raws keys should be sacrificed to enhance the total privacy. First, we introduce the notation used in the proof.

The qudit system $A$ considered in Appendix \ref{Sc:Proof} is on an optical mode, whose creation operator is $a^\dag$ and Hilbert space is denoted as $\mathcal{H}^A$. A $k$-photon Fock state $\ket{k}_A$ is defined as
\begin{equation} \label{eq:defFock}
\ket{k}_A \equiv \dfrac{(a^\dag)^k}{\sqrt{k!}} \ket{0}_A,
\end{equation}
where $\ket{0}_A$ is the vacuum state. A coherent state $\ket{\alpha}_A$ is defined as
\begin{equation}
\ket{\alpha}_A \equiv e^{-\frac{1}{2}|\alpha|^2} \sum_{k=0}^{\infty} \dfrac{\alpha^k}{\sqrt{k!}} \ket{k}_A
\end{equation}
The photon number of $\ket{\alpha}_A$ follows a Poisson distribution $P_\mu(k) = e^{-\mu}\mu^k/k!$, where $\mu = |\alpha|^2$ is the mean photon number, also the light intensity. In the following proof, we specifically select
\begin{equation}
U=U_A=U_B=e^{i\pi a^\dag a}
\end{equation}
, which satisfies the condition $U^2=I$. When applied on a Fock state $\ket{n}$, this operation adds an additional phase $(-1)^n$, which has no effect on even-photon Fock states, while changes the phase of odd-photon Fock states. As a result, we have the folowing corollary.

\begin{corollary}
A state $\rho$ on two modes $A,B$ is an odd (even) state with respect to $U_A\otimes U_B$ iff $\hat{\Pi}_{o} \rho \hat{\Pi}_{o} = \rho$ ($\hat{\Pi}_{e} \rho \hat{\Pi}_{e} = \rho$). $\hat{\Pi}_{o}$ and $\hat{\Pi}_{e}$ are the projectors of the odd and even subspace respectively, which are defined as
\begin{equation}
\begin{aligned}
\hat{\Pi}_{o} &= \sum_{m + n: odd} \ket{m,n}_{AB} \bra{m,n}, \\
\hat{\Pi}_{e} &= \sum_{m + n: even} \ket{m,n}_{AB} \bra{m,n}, \\
\end{aligned}
\end{equation}
 where $\{\ket{m,n}_{A,B}\}_{m,n}$ are the Fock basis on optical modes $A$ and $B$.
\end{corollary}

\subsection{Phase randomization and parity state}

To generate a mixture of parity states, we introduce the ``twirling'' operation,
\begin{equation}
\mathcal{E}(\sigma_{AB}) = \dfrac{1}{2} \sigma_{AB} + \dfrac{1}{2} (U_A \otimes U_B) \sigma_{AB} (U_A \otimes U_B)^\dagger,
\end{equation}

Note that $\mathcal{E}(\sigma_{AB})$ keeps invariant under the transformation $U_A\otimes U_B$, and therefore is a mixture of parity state for any state $\sigma_{AB}$ on two optical modes $A$ and $B$ according to Corollary \ref{prop:mixparity}. To implement $\mathcal{E}$ in PM-QKD, Alice and Bob can first randomize their systems $A,B$ with random phase $\phi_a$ and $\phi_b$ independently, announce the random phase, and post-select the runs with the same random phase \emph{after Eve's detection announcements}. PM-QKD with independant randomization is summarized as Protocol II below.

\textbf{\uline{Protocol II}}
\begin{enumerate}
\item 
State preparation: Alice and Bob share a known state $\sigma_{AB}$ on two optical mode $A,B$, at the beginning of each run. They initialize their qubits $A', B'$ in $\ket{+i}$. Alice applies the control gate $C_{A'A}(U)$ to qubit $A'$ and optical pulse $A$, and adds an extra random $0/\pi$ phase $\phi_a$ on $A$. Similarly, Bob applies $C_{B'B}(U)$, $\phi_b$ to $B'$ and $B$.


\item 
Measurement: Alice and Bob send the two optical pulses $A$ and $B$ to an untrusted party, Eve, who is supposed to perform joint measurement and obtain the detection results $L$ or $R$, which is expected to be projective measuremnt on the basis $\mathcal{E}(\sigma_{AB})$ and $(U_A\otimes I) \mathcal{E}(\sigma_{AB}) (U_A\otimes I)^\dag$. However, this is no an assumption or requirement of the measurement of Eve, as an untrusted party.

\item 
Announcement: Eve announces the detection result or no successful detection for each round. After that, Alice and Bob announce their random phases $\phi_a$ and $\phi_b$.

\item 
Sifting: When Eve announces an $L$ or $R$ click, Alice and Bob keep the qubits of systems $A'$ and $B'$. In addition, Bob applies a Pauli $Y$-gate to his qubit $B'$ if Eve's announcement is $R$ click. If $|\phi_a - \phi_b| = \pi$, Bob applies another Pauli $Y$-gate on $B'$.

Alice and Bob perform the above steps for many rounds and end up with a joint $2n$-qubit state $\rho_{A' B'}\in(\mathcal{H}_{A'}\otimes\mathcal{H}_{B'})^{\otimes n}$.

\item 
Parameter estimation: Alice and Bob estimate the click ratios caused by even state fractions.

\item 
Key generation: Alice and Bob perform local $Z$ measurements on $\rho_{A' B'}$ to generate raw data string $\kappa_A$ and $\kappa_B$. They reconcile the key string to $\kappa_{A}$ by an encrypted classical channel, with the consummation of $l_{ec}$-bit keys. After that, they perform privacy amplification according to the even state ratio to generate keys.
\end{enumerate}

In Protocol II, the signals with $|\phi_a - \phi_b| = \pi$ can also be used if Bob performs extra $Y_{B'}$ gate on system $B'$ before the key generation step \cite{ma2018phase}. Besides the independant randomization, the Protocol II is the same as Protocol I, which indicates that the security of Protocol II can be reduced to Protocol I, and the information leakage in Protocol II can be bounded by the phase error rate, $E^X_o$, which is introduced in Protocol Ib, an weaker form of Protocol I. The problem of independent randomization is analyzed in the following text.

In the protocol, Alice and Bob have to announce $\phi_a, \phi_b$ publicly to post-select the runs with the same random phase $\phi_a = \phi_b$. As is analyzed in Ref.~\cite{ma2018phase}, the random phase $\phi_a, \phi_b$ is correlated to the key information, which can be utilized by Eve. To model the information leakage caused by random phase announcement, consider the case where Alice and Bob share a pure state $\ket{\psi}_{AB}$. In a purified scenario, suppose a qubit register $P$ is initialized in the state $\ket{+}$, then Alice and Bob realize $\mathcal{E}$ with
\begin{equation}
\begin{aligned}
\ket{\Psi}_{PAB} &= U(\ket{+}_P \ket{\psi}_{AB})\\
 &= \dfrac{1}{\sqrt{2}} ( \ket{+}_P \ket{\psi}_{AB} + \ket{-}_P (U_A(\pi) \otimes U_B(\pi)) \ket{\psi}_{AB} ), \\
&= \ket{0}_P \ket{\bar{\psi_e}}_{AB} + \ket{1}_P \ket{\bar{\psi_o}}_{AB}\\
&= \sqrt{p_e} \ket{0}_P \ket{\psi_e}_{AB} + \sqrt{p_o} \ket{1}_P \ket{\psi_o}_{AB} ,
\end{aligned}
\end{equation}
where
\begin{equation} \label{eq:psiepsio}
\begin{aligned}
\ket{\bar{\psi_e}}_{AB} &= \dfrac{1}{2} (\ket{\psi}_{AB} + (U_A \otimes U_B) \ket{\psi}_{AB}), \\
\ket{\bar{\psi_o}}_{AB} &= \dfrac{1}{2} (\ket{\psi}_{AB} - (U_A \otimes U_B) \ket{\psi}_{AB}).
\end{aligned}
\end{equation}
Here $\ket{\bar{\psi_e}}_{AB}, \ket{\bar{\psi_o}}_{AB}$ are unnormalized even and odd state, where $p_e = \braket{\bar{\psi_e}|\bar{\psi_e}}$ and $p_o = \braket{\bar{\psi_o}|\bar{\psi_o}}$.

Therefore, the register $P$ records whether $\pi$-modulation is applied by $X$-basis state $\ket{\pm}$, while the parity information of the state is kept in the $Z$-basis state $\ket{0/1}$ of the register, as shown in the equation given above.

In this scenario, phase announcement can be interpreted as the following process: Alice and Bob prepare $\ket{\Psi}_{PAB}$, as a purification of $\mathcal{E}(\sigma_{AB})$, and measure system $P$ on $X$-basis, followed by announcing the result to Eve after the detection announcement.
In a worse case, we can reduce the protocol to Protocol Ib, where Eve holds the system $P$ after her detection announcement.

To conclude, PM-QKD protocol can be realized by any initial state $\sigma_{AB}$ with the assist of a phase randomization, at the expense of losing the capability of distinguishing parity components.

\begin{figure*}[htbp]
\centering
\includegraphics[width=9cm]{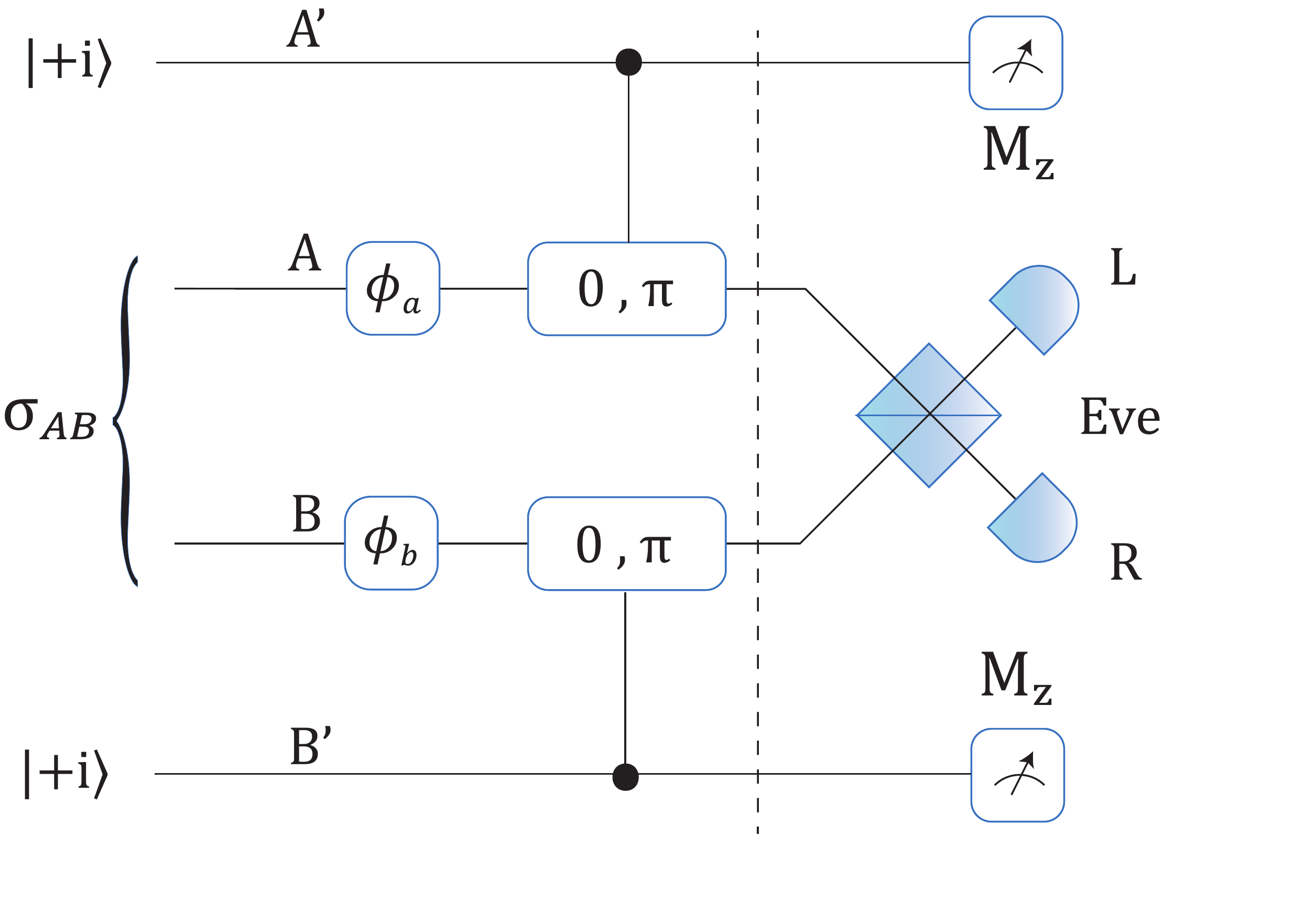}
\caption{Realistic PM-QKD protocol with extra $0/\pi$-phase randomization. $\sigma_{AB}$ is a generic state on optical modes $A$ and $B$. $\phi_a,\phi_b$ are two random phases which is either $0$ or $\pi$. In practice, the phase can be absorbed into the controlled-$U(\pi)$ operations afterwards.} \label{fig:Pro2}
\end{figure*}

\subsection{PM-QKD based on coherent state}

So far, we don't make any assumption on the structure of initial state $\sigma_{AB}$. Here, we focus on one specific implementation. Due to the fact that coherent states are easy to implement in experiments, we set $\sigma_{AB}$ as $\ket{\sqrt{\mu_a}}_A\otimes\ket{\sqrt{\mu_b}e^{i\delta}}_B$. After the twirling phase randomization, the state becomes
\begin{widetext}
\begin{equation} \label{eq:muamub0pi}
\begin{aligned}
\rho &= \dfrac{1}{2} ( \ket{\sqrt{\mu_a}}_A\bra{\sqrt{\mu_a}}\otimes\ket{\sqrt{\mu_b}e^{i\delta}}_B\bra{\sqrt{\mu_b}e^{i\delta}} + \ket{-\sqrt{\mu_a}}_A\bra{-\sqrt{\mu_a}}\otimes\ket{-\sqrt{\mu_b}e^{i\delta}}_B\bra{-\sqrt{\mu_b}e^{i\delta}} ) \\
&=  p_{even} \ket{\psi_e^\delta(\mu_a,\mu_b)}_{AB}\bra{\psi_e^\delta(\mu_a,\mu_b)} + p_{odd} \ket{\psi_o^\delta(\mu_a,\mu_b)}_{AB}\bra{\psi_o^\delta(\mu_a,\mu_b)},
\end{aligned}
\end{equation}
\end{widetext}
where
\begin{widetext}
\begin{equation} \label{eq:psidelta}
\begin{aligned}
\ket{\psi_e^\delta(\mu_a,\mu_b)}_{AB} &= \dfrac{1}{\sqrt{p_{even}}} ( \ket{\sqrt{\mu_a}}_A\otimes\ket{\sqrt{\mu_b}e^{i\delta}}_B + \ket{-\sqrt{\mu_a}}_A\otimes\ket{-\sqrt{\mu_b}e^{i\delta}}_B ), \\
\ket{\psi_o^\delta(\mu_a,\mu_b)}_{AB} &= \dfrac{1}{\sqrt{p_{odd}}} ( \ket{\sqrt{\mu_a}}_A\otimes\ket{\sqrt{\mu_b}e^{i\delta}}_B - \ket{-\sqrt{\mu_a}}_A\otimes\ket{-\sqrt{\mu_b}e^{i\delta}}_B ), \\
\end{aligned}
\end{equation}
\end{widetext}
Here the probabilities $p_{even}, p_{odd}$ are the normalization factors.
In the case where $\delta=0$, the initial state is an unbiased mixing of $\ket{\sqrt{\mu_a}}_A\otimes\ket{\sqrt{\mu_b}}_B$ and $\ket{-\sqrt{\mu_a}}_A\otimes\ket{-\sqrt{\mu_b}}_B$, in which case the even and odd components is
\begin{widetext}
\begin{equation} \label{eq:rhoerhoo}
\begin{aligned}
\rho_e(\mu_a,\mu_b) &= \dfrac{1}{2}(\ket{\psi_e^0(\mu_a,\mu_b)}_{AB}\bra{\psi_e^0(\mu_a,\mu_b)} + \ket{\psi_e^\pi(\mu_a,\mu_b)}_{AB}\bra{\psi_e^\pi(\mu_a,\mu_b)}), \\
\rho_o(\mu_a,\mu_b) &= \dfrac{1}{2}(\ket{\psi_o^0(\mu_a,\mu_b)}_{AB}\bra{\psi_o^0(\mu_a,\mu_b)} + \ket{\psi_o^\pi(\mu_a,\mu_b)}_{AB}\bra{\psi_o^\pi(\mu_a,\mu_b)}). \\
\end{aligned}
\end{equation}
\end{widetext}
 Note that $\rho_e(\mu_a,\mu_b) $ ($
\rho_o(\mu_a,\mu_b) $) is only comprised of even (odd)-phton Fock states, and therefore in even (odd) subspace.

In the final key distillation steps, we deal with the runs with $\phi_a = \phi_b$ and $|\phi_a - \phi_b| = \pi$ together, and the overall phase error rate $E^{ph}$ is given by the fraction of clicks caused by even components $\rho_e(\mu_a,\mu_b)$. To estimate the even clicks, we only need to estimate the yield $Y_{even}$, i.e. the detection probability when Alice and Bob send the state $\rho_e(\mu_a,\mu_b)$. The phase error rate $E^{ph}$ is
\begin{equation}
E^{ph} = \dfrac{p_{even}Y_{even}}{Q_{\mu_a, \mu_b}}.
\end{equation}

Here, the only task remained is to estimate $Y_{even}$, as $Q_{\mu_a, \mu_b}$ given experiment parameters and data. $p_{even}$ describes the proportion of even photon rounds, which related to the intensities $\mu_a, \mu_b$.

\subsection{Phase error estimation and continuous randomization}
As a common technique, decoy state method \cite{Hwang2003Decoy,Lo2005Decoy,Wang2005Decoy} can be applied to estimate the value of $Y_{even}$. The core of decoy state method is to use a set of testing states to learn Eve's behavior on specific components in the signal states. The decoy is based on the observation that, if the same components appears in both the (mixed) signal states and the (mixed) testing states, Eve cannot attack on them with different manners in principle. In this section, we consider the asymptotic case with $n\to\infty$, so that there is no statistical fluctuation. The finite-size analysis is in Appendix \ref{Sc:finite}.

For the simplicity of discussion, we assume that $\mu_a = \mu_b = \mu/2$. The following argument can be easily generalize to the cased when $\mu_a \neq \mu_b$. Denote the fraction of odd- and even-labeled clicked rounds as
\begin{equation}
q_{odd} = \dfrac{n_{odd}}{n}; q_{even} = \dfrac{n_{even}}{n}.
\end{equation}
According to Eq.~\eqref{eq:rhoerhoo}, the signal source is combined by $\ket{\psi_e}$ and $\ket{\psi_o}$ with probability of $p_{even}$ and $p_{odd}$,
\begin{equation}
\rho = p_{odd}^{\mu}\rho_o(\frac{\mu}{2},\frac{\mu}{2}) + p_{even}^{\mu}\rho_e(\frac{\mu}{2},\frac{\mu}{2}),
\end{equation}
Here we add superscript $\mu$ to the probabilities and states, since they are the functions of $\mu$. The core is to estimate the detection probability of $\rho_e(\frac{\mu}{2},\frac{\mu}{2})$, namely, the yield $Y^{\mu}_{even}$.

The fraction of odd- and even-parity states in the final detected signal is given by
\begin{equation}
\begin{aligned}
\label{eqn:q parity}
q_{odd}^\mu &= p_{odd}^\mu \dfrac{Y_{odd}^\mu}{Q_\mu}, \\
q_{even}^\mu &= p_{even}^\mu \dfrac{Y_{even}^\mu}{Q_\mu},
\end{aligned}
\end{equation}
where $Q_\mu$ is the total gain of the signals. For signals with intensity $\mu$, we have
\begin{equation}
\label{eqn:decoy parity}
Q^{\mu} = p_{odd}^{\mu} Y^{\mu}_{odd} + p_{even}^{\mu} Y^{\mu}_{even}. \\
\end{equation}

To efficiently estimate $q_{odd}$ and $q_{even}$, Alice and Bob adjust the intensity $\mu$ of their prepared coherent lights. With constraints from different intensities,
\begin{equation}
\begin{aligned}
\label{eqn:decoy parity 2}
Q^{\mu} = p_{odd}^{\mu} Y^{\mu}_{odd} + p_{even}^{\mu} Y^{\mu}_{even}, \\
Q^{\nu} = p_{odd}^{\nu} Y^{\nu}_{odd} + p_{even}^{\nu} Y^{\nu}_{even}, \\
\end{aligned}
\end{equation}
we have better estimation of $Y^\mu_{odd}$ and $Y^\mu_{even}$. Note that $Y^{\mu}_{odd}, Y^{\mu}_{even}$ is dependent on the intensity $\mu$. For different signal intensities $\mu,\nu$, the difference of yield $Y^{\mu}_{odd}, Y^{\mu}_{even}$ is bounded by
\begin{equation} \label{eq:Ymunu}
    |Y^\mu_{odd} - Y^\nu_{odd}| \leq \sqrt{1 - F^2_{\mu\nu}},
\end{equation}
where $F_{\mu\nu}$ is the fidelity between the odd state $\ket{\psi_o}$ with intensities $\mu$ and $\nu$, respectively,
\begin{equation}
    F_{\mu\nu} = Tr[\sqrt{\sqrt{\rho_{odd}^\mu}\rho_{odd}^\nu\sqrt{\rho_{odd}^\mu}}].
\end{equation}

In order to obtain a tighter bound on the estimation of $Y^{\mu}_{even}$, we can introduce extra phases other than $\{0,\pi\}$ for coherent states $\ket{\sqrt{\mu/2}}_{A}\otimes \ket{\sqrt{\mu/2}}_{B}$. As an extreme case, one choice is to randomize the phase continuously from $0$ to $2\pi$, i.e. continuous randomization. Note that, for a coherent state,
\begin{equation} \label{eqn:phase random}
\dfrac{1}{2\pi}\int_0^{2\pi} d\phi \ket{\sqrt{\mu} e^{i\phi}} \bra{\sqrt{\mu} e^{i\phi}} = \sum_{k=0}^{\infty} P_\mu(k) \ket{k}\bra{k},
\end{equation}
where $P_\mu(k)=e^{-\mu}\frac{\mu^k}{k!}$ is the proportion of Fock state $\ket{k}$ in the mixed states.

If we randomize the phase of two coherent pulses simultaneously, for states with $|\phi_a - \phi_b| = \delta$, we have
\begin{widetext}
\begin{equation} \label{eqn:two phase random}
\dfrac{1}{2\pi}\int_0^{2\pi} d\phi \ket{\sqrt{\mu/2} e^{i\phi}}_A \bra{\sqrt{\mu/2} e^{i\phi}} \otimes \ket{\sqrt{\mu/2} e^{i(\phi + \delta)}}_B \bra{\sqrt{\mu/2} e^{i(\phi + \delta)}} = \sum_{k=0}^{\infty} P_\mu(k) \ket{\bar{k}^\delta}_{AB}\bra{\bar{k}^\delta},
\end{equation}
\end{widetext}
the $k$-photon state $\ket{\bar{k}^\delta}_{AB}$ is
\begin{equation} \label{eq: kAB}
\ket{\bar{k}^\delta}_{AB} = \dfrac{(a^\dag + e^{i\delta} b^\dag)^k}{\sqrt{2^k k!}} \ket{00}_{AB}.
\end{equation}

Consider the simultaneous randomization and the key encoding process, the total phase difference of Alice and Bob's coherent state can be $\phi_a = \phi_b$ or $|\phi_a - \phi_b| = \pi$, which indicates the mixed $k$-photon state, which is sent to Eve, is
\begin{equation}
\rho_k = \dfrac{1}{2}(\ket{\bar{k}^0}_{AB}\bra{\bar{k}^0} + \ket{\bar{k}^\pi}_{AB}\bra{\bar{k}^\pi}),
\end{equation}
which is independent of the intensity $\mu$. In this case, the state Alice and Bob send out can be regarded as a probabilistic mixture of mixed $k$-photon state $\rho_k$. By directly applying Lemma \ref{Lem:parity}, we can estimate the phase error by
\begin{equation} \label{eqn:EX Fock}
\begin{aligned}
E^{ph} &=  \sum_{k=0}^{\infty} q_{2k}, \\
\end{aligned}
\end{equation}
where $q_k$ is the proportion of detection event caused by state $\rho_{AB}$, especially $q_0$ corresponds to the vacuum signal detection, i.e.~dark counts.

The source components are Fock states $\{\ket{k}\}$, whose yields $\{Y_k\}$  are independent of $\mu$. The fractions $q_k^\mu$ of $k$-photon component in the final detected signals are given by
\begin{equation}
\label{eqn:q Fock}
q^{\mu}_k = P^\mu(k) \dfrac{Y_k}{Q_\mu}.
\end{equation}
The overall gain is given by
\begin{equation}
\begin{aligned}
\label{eqn:decoy Fock}
Q_\mu &= \sum_{k=0}^{\infty} P^\mu(k) Y_k, \\
\end{aligned}
\end{equation}
Similarly, we use decoy state method in continuous randomization case, and there is a set of \eqref{eqn:decoy Fock} with different signal intensities $\mu$ and corresponding proportion $P^\mu(k)$, which can be used to bound $\{Y_k\}$, therefore we can estimate $E^{ph}$ efficiently.

The continuous phase-randomized protocol, named Protocol III, is as follows,

\textbf{\uline{Protocol III}}
\begin{enumerate}
\item 
State preparation: Alice and Bob prepares coherent states $\ket{\sqrt{\mu_a}}_A\otimes \ket{\sqrt{\mu_b}}_B$ on two optical modes $A,B$, at the beginning of each run. They initialize their qubits $A', B'$ in $\ket{+i}$. Alice applies the control gate $C_{A'A}(U)$ to qubit $A'$ and optical pulse $A$, and add an extra random $0\sim 2\pi$ phase $\phi_a$ on $A$. Similarly, Bob applies $C_{B'B}(U)$, $\phi_b$ to $B'$ and $B$.


\item 
Measurement: Alice and Bob send the two optical pulses $A$ and $B$ to an untrusted party, Eve, who is supposed to perform joint measurement and obtain the detection results $L$ or $R$.

\item 
Announcement: Eve announces the detection result or no successful detection for each round. After that, Alice and Bob announce their random phases and intensity settings $\{\phi_a,\mu_a\}$ and $\{\phi_b,\mu_b\}$ and keep the signals with $|\phi_a - \phi_b| = 0/\pi$ and $\mu_a = \mu_b$.

\item 
Sifting: When Eve announces an $L/R$ click, Alice and Bob keep the qubits of systems $A'$ and $B'$. In addition, Bob applies a Pauli $Y$-gate to his qubit $B'$ if Eve's announcement is $R$ click. If $|\phi_a - \phi_b| = \pi$, Bob applies another Pauli $Y$-gate on $B'$.

Alice and Bob perform the above steps for many rounds and end up with a joint $2n$-qubit state $\rho_{A' B'}\in(\mathcal{H}_{A'}\otimes\mathcal{H}_{B'})^{\otimes n}$.

\item 
Parameter estimation: Alice and Bob estimate the click ratio caused by even state fractions by decoy-state methods.

\item 
Key generation: For the signals with intensity $\mu_a = \mu_b = \mu/2$, Alice and Bob perform local $Z$ measurements on $\rho_{A' B'}$ to generate raw data string $\kappa_A$ and $\kappa_B$. They reconcile the key string to $\kappa_{A}$ by an encrypted classical channel, with the consummation of $l_{ec}$-bit keys. After that, they perform privacy amplification according to the even state ratio to generate keys.
\end{enumerate}

\subsection{Practical issues in PM-QKD} \label{Sc:practical}

In practice, to bound the phase error, we only need to bound the yield of single-photon components $\rho_1$,
\begin{equation}
\rho_1 = \dfrac{1}{2} (\ket{01}_{AB}\bra{01} + \ket{10}_{AB}\bra{10}),
\end{equation}
and the overall odd phase error rate is bounded by
\begin{equation} \label{eq:Eph}
E^{ph} \leq 1 - q_1,
\end{equation}
where $q_1 = P_\mu(1) Y_1/Q_\mu$.

To make the above Protocol III practical, we now consider the following issues.
\begin{enumerate}
\item From continuous phase randomization to discrete phase randomization.

A continuous phase randomization and phase post-selection is practically intractable. However, in practice, randomizing the phases of coherent pulses discretely is enough. For a coherent state $\ket{\sqrt{\mu}}$, if we randomize its phase discretely with $\{\phi_j = \frac{2\pi}{D}j \}_{j=0}^{D-1}$, the state can be expanded by a group of ``pseudo'' Fock state\cite{Cao2015Discrete},
\begin{equation} \label{eq:discrete phase random}
\dfrac{1}{D} \sum_{j=0}^{D-1} \ket{\sqrt{\mu}e^{i\phi_j}}_C\bra{\sqrt{\mu}e^{i\phi_j}} = \sum_{k=0}^{D-1} P^\mu_D(k) \ket{\lambda_k}_C\bra{\lambda_k},
\end{equation}
where
\begin{equation}
\begin{aligned}
\ket{\lambda_k}_C &= \frac{e^{-\mu/2}}{\sqrt{P^\mu(k)}}\sum_{l=0}^{\infty} \dfrac{(\sqrt{\mu})^{lD+k}}{\sqrt{(lD+k)!}}\ket{lD+k}_C, \\
P^\mu_D(k) &= \sum_{l=0}^{\infty} \dfrac{\mu^{lD+k}e^{-\mu}}{(lD+k)!},
\end{aligned}
\end{equation}
as we can see, when $D$ becomes large, $\ket{\lambda_0}$ and $\ket{\lambda_1}$ will get close to the Fock state $\ket{0}$ and $\ket{1}$.

Now we calculate the deviation of $q_1$ caused by discrete phase randomization. Without loss of generality, we set $|\phi_a - \phi_b| = \delta$. After the discrete phase randomization, the state is
\begin{widetext}
\begin{equation} \label{eq:discrete phase random2}
\dfrac{1}{D} \sum_{j=0}^{D-1} \ket{\sqrt{\frac{\mu}{2}} e^{i\phi_j}}_A \bra{\sqrt{\frac{\mu}{2}} e^{i\phi_j}} \otimes \ket{\sqrt{\frac{\mu}{2}} e^{i(\phi_j + \delta)}}_B \bra{\sqrt{\frac{\mu}{2}} e^{i(\phi_j + \delta)}} = \sum_{k=0}^{\infty} P^\mu_D(k) \ket{\bar{\lambda}_k^\delta}_{AB}\bra{\bar{\lambda}_k^\delta},
\end{equation}
\end{widetext}
the $k$-photon state $\ket{\bar{\lambda}_k^\delta}_{AB}$ is
\begin{equation}
\ket{\bar{\lambda}_k^\delta}_{AB} = \frac{e^{-\mu/2}}{\sqrt{P^\mu(k)}}\sum_{l=0}^{\infty} \dfrac{(\sqrt{\mu})^{lD+k}}{\sqrt{(lD+k)!}} \ket{\overline{lD + k}^\delta}_{AB},
\end{equation}
where $\ket{\bar{k}^\delta}_{AB}$ is defined in Eq.~\eqref{eq: kAB}. 

We compare the fidelity between $\ket{\bar{1}^\delta}_{AB}$ and $\ket{\bar{\lambda}_1^\delta}_{AB}$,
\begin{widetext}
\begin{equation}
\begin{aligned}
|\braket{\bar{1}^\delta|\bar{\lambda}_1^\delta}|^2 &= \frac{e^{-\mu}}{P^\mu_D(1)} \left|\sum_{l=0}^\infty \dfrac{(\sqrt{\mu})^{lD+1}}{(lD+1)!} \braket{\bar{1}^\delta|\overline{lD+1}^\delta} \right|^2 \\
&= \frac{e^{-\mu}}{P^\mu_D(1)} \mu \\
&= \frac{ e^{-\mu} \mu }{ e^{-\mu} (\mu + \frac{\mu^{(D+1)}}{(D+1)!} + \frac{\mu^{(2D+1)}}{(2D+1)!} + ...) } \\
&= \frac{1}{1 + \frac{\mu^{D}}{(D+1)!} + \frac{\mu^{2D}}{(2D+1)!} +... } \\
&\geq 1 - \frac{\mu^D}{(D+1)!},
\end{aligned}
\end{equation}
\end{widetext}
the final inequality holds because $((n+1)D+1)! \geq (nD+1)! (D+1)!$ for $n\geq 1$. Note that the fidelity is independent of the phase difference $\delta$.

Therefore, according to Eq.~\eqref{eq:Ymunu}, the yield difference of $\ket{\tilde{1}^\delta}$ and $\ket{\bar{\lambda}_1^\delta}_{AB}$ is bounded by
\begin{equation}
|Y_{1} - Y^\mu_{\lambda_1}| \leq \sqrt{1 - |\braket{\bar{1}^\delta|\bar{\lambda}_1^\delta}|^2} \leq \sqrt{\frac{\mu^{D}}{(D+1)!}},
\end{equation}
then the difference between estimated $\ket{\bar{\lambda}_1^\delta}_{AB}$ fraction, denoted by $q_1$ and $q_{\lambda_1}^\mu$, is bounded by
\begin{equation}
|q_1 - q_{\lambda_1}^\mu| \leq P^\mu_D(1) \frac{|Y_{1} - Y^\mu_{\lambda_1}|}{Q_\mu} \leq \frac{P^\mu_D(1)}{Q_\mu} \sqrt{\frac{\mu^{D}}{(D+1)!}} = \frac{\xi_D(\mu)}{Q_\mu},
\end{equation}
note that $\xi_D(\mu) \equiv P^\mu_D(1) \sqrt{\frac{\mu^{D}}{(D+1)!}} \approx \frac{\mu^{D/2 + 1} e^{-\mu}}{\sqrt{(D+1)!}}$, which is only correlated to the signal intensity $\mu$ and discrete phase number $D$. Therefore, we can compare the ratio $\xi_D(\mu)/Q_\mu$, which illustrates the deviation of the estimated $q_1$ from the real $q_{\lambda_1}^\mu$.

If we set $D=16$, then $\xi_D(\mu) \approx \frac{\mu^9 e^{-\mu}}{\sqrt{17!}}$, which is a tiny value when $\mu<1$, comparing to the gain $Q_\mu \approx \eta\mu$, where $\eta$ is the channel transmittance from Alice or Bob to the middle point Eve. Therefore, we can safely ignore the effect caused by discrete phase randomization and borrow the former decoy-state method based on continuous phase randomization \cite{Lo2005Decoy,Ma2005Practical}.

\item Key generation and parameter estimation with the signals which the phases are not aligned.

In Protocol III, only the states with aligned phases $|\phi_a - \phi_b|=0, \pi$ and the same intensity $\mu_a = \mu_b$ will be post-selected to estimate the detections caused by single-photon components and generate keys. This will cause a huge waste on sifting in a practical finite-size case. We now discuss how to use this signals for key generation and parameter estimation.

First, we notice that for signals with $\mu_a = \mu_b=\mu/2$ and $|\phi_a - \phi_b| = \delta, \delta+\pi$, we can regard it as a specific case of Protocol III, where Alice and Bob originally share a state $\ket{\sqrt{\mu_a}}_A \otimes \ket{\sqrt{\mu_b}e^{i\delta}}_B$. According to Eq.~\eqref{eq: kAB}, the mixed $k$-photon state when $|\phi_a - \phi_b| = \delta, \delta+\pi$ is
\begin{equation} \label{eq:rhokdelta}
\rho_k^\delta = \dfrac{1}{2}(\ket{\bar{k}^\delta}_{AB}\bra{\bar{k}^\delta} + \ket{\bar{k}^{\delta + \pi}}_{AB}\bra{\bar{k}^{\delta + \pi}}).
\end{equation}

In general, the $k$-photon state is correlated with the misaligned phase $\delta$, which implies that states with different mismatched phases have different mixed k-photon states. However, this state is not true for $k-1$, where the single-photon state
\begin{equation}
\rho_1^\delta = \dfrac{1}{2}(\ket{01}_{AB}\bra{01} + \ket{10}_{AB}\bra{10}),
\end{equation}
is independent of the misaligned phase. Therefore, states with different unaligned phases has the same single-photon component, which indicates that we can use all of states to estimate the yield of single-photon component, regardless of the misaligned phase $\delta$. The phase error can be bounded by
\begin{equation} \label{eq:Ephq1}
E^{ph} \leq 1 - q_1.
\end{equation}

Based on this observation, during the post-processing, Alice and Bob first reconcile their sifted raw key bits $K_a$ and $K_b$ with for each group $j_s = j_a - j_b$ separately. If the error rate in a group of data is too large, they can simply discard that group. Denote the group set $J$ to be the set of remaining phase group indexes $\{j_s\}$. That is, if $j_s \in J$, then the phase group $j_s$ is kept for key generation. They then estimate the even photon fraction $q_{even}$ for all the remaining data in $J$ and perform privacy amplification.

In a more general scenario, Alice and Bob can estimate the yield of state $\ket{01}_{AB}$ and $\ket{10}_{AB}$, denoted as $Y_{01}$ and $Y_{10}$, respectively. The overall yield $Y_1$ can be calculated by
\begin{equation}
Y_1 = \frac{1}{2} ( Y_{01} + Y_{10}).
\end{equation}

Similar to the traditional MDI-QKD \cite{Lo2012Measurement,ma2012statistical,curty2014finite}, Alice and Bob can use signals with different intensities to achieve a better estimation on $Y_1$. However, in this article, to simplify the discussion, we will focus the signals with $\mu_a = \mu_b$, and leave the tighter estimation for future works.

\item From infinite decoy state analysis to finite decoy state setting.

In practice, with finite data size, we can use only finite rounds for testing. To accurately bound the even state components, we have to use the testing states with finite intensity settings.

We will explicitly analyze this problem in Appendix~\ref{Sc:finite}. As a result, we will show that Alice and Bob can use only vacuum + weak decoy state, similar to the BB84 decoy state analysis \cite{Ma2005Practical}.

\end{enumerate}

With all the factors taken into consideration, the practical version of Protocol III, named Protocol IV, is presented as below.

\textbf{\uline{Protocol IV}}
\begin{enumerate}
\item 
State preparation: Alice prepares coherent state $\ket{\sqrt{\mu_a}e^{i\phi_{j_a}}}_A$ on optical mode $A$, with $\mu_a \in \{0,\nu/2, \mu/2\}$, and $\phi_{j_a} \in \{ \frac{2\pi}{D} j_a \}_{j_a=0}^{D-1}$. She initializes her qubit $A'$ in $\ket{+i}$. She applies the control gate $C_{A'A}(U)$ to qubit $A'$ and optical pulse $A$. Similarly, Bob prepares $\ket{\sqrt{\mu_b}e^{i\phi_{j_b}}}_B$ on optical mode $B$. He initializes his qubit $B'$ in $\ket{+i}$. He applies the control gate $C_{B'B}(U)$ to qubit $B'$ and optical pulse $B$.

\item 
Measurement: Alice and Bob send their optical pulses, $A$ and $B$, to an untrusted party, Eve, who is expected to perform an interference measurement and record the detector ($L$ or $R$) that clicks.

\item 
Announcement: Eve announces the detection result or no successful detection for each round. After that, Alice and Bob announce their random phases and intensity settings $\{j_a,\mu_a\}$ and $\{j_b,\mu_b\}$ and keep the signals with $\mu_a = \mu_b$.

\item 
Sifting: When Eve announces an $L/R$ click, Alice and Bob keep the qubits of systems $A'$ and $B'$. In addition, Bob applies a Pauli $Y$-gate to his qubit $B'$ if Eve's announcement is $R$ click. If $\frac{3D}{4} > [(j_a - j_b) \text{ mod } D ] \geq \frac{D}{4}$, Bob applies another Pauli $Y$-gate on $B'$.

Alice and Bob perform the above steps for many rounds and end up with a joint $2n$-qubit state $\rho_{A' B'}\in(\mathcal{H}_{A'}\otimes\mathcal{H}_{B'})^{\otimes n}$.

\item 
Parameter estimation: For all the raw data that they have retained, Alice and Bob record the detect number $M_{i_a, i_b}^{(j_s)}$ of different intensity combinations $\{\mu_{i_a}, \mu_{i_b}\}$ and phase groups $j_s$. They then estimate the information leakage $E^{ph}_{\mu}$ using Eq.~\eqref{eq:Eph}.

\item 
Key generation: For the signals with intensity $\mu_a = \mu_b = \mu/2$, Alice and Bob perform local $Z$ measurements on $\rho_{A' B'}$ to generate raw data string $\kappa_A$ and $\kappa_B$. They first group the signals by $j_s = \min[(j_a - j_b)\text{ mod } D, (j_b - j_a)\text{ mod } D]$. After that, they keep the signals with $j_s$ in a set $J$. They reconcile the corresponding key string to $\kappa_{A}$ by an encrypted classical channel, with the consummation of $l_{ec}^{j_s}$-bit keys according to $j_s$. After that, they perform privacy amplification according to the estimated single-photon ratio $q_1$ to generate keys.
\end{enumerate}

The key rate of this protocol is
\begin{equation}
r = \dfrac{2}{D} \sum_{j_s\in J} \left[ 1 - H(E^{ph}) - l^{j_s}_{ec} \right].
\end{equation}
Typically, $l_{ec}^{j_s}$ can be replaced by $fH(E^{j_s})$, where $f$ is error correction efficiency and $E^{j_s}$ is bit error rate for each phase group $j_s$ that is in $J$. $E^{ph}$ is the phase error rate bounded by Eq.~\eqref{eq:Ephq1}.

Following Shor and Preskill \cite{Shor2000Simple}, we can move the measurement before the sifting step, in which case $C-\pi$ rotation is replaced by the bit flip operation. Then Protocol IV reduces to PM-QKD protocol in the manuscript.

\textbf{\uline{PM-QKD}}
\begin{enumerate}
\item
\textbf{State preparation}:
Alice randomly generates a key bit $\kappa_a$, and picks a random phase $\phi_{j_a}$ from the set $\{j_a \dfrac{2\pi}{D}\}_{j_a =0}^{D-1}$, and the intensity $\mu_{i_a}$ from $\{0,\nu/2,\mu/2\}$. She then prepares the coherent state $\ket{\sqrt{\mu_{i_a}}e^{i(\phi_{j_a}+\pi\kappa_a)}}_{A}$. Similarly, Bob generates $\kappa_b$, $\phi_{j_b}$, $\mu_{i_b}$ and then prepares $\ket{\sqrt{\mu_{i_b}}e^{i(\phi_{j_b}+\pi\kappa_b)}}_{B}$.

\item
\textbf{Measurement}: Alice and Bob send their optical pulses, $A$ and $B$, to an untrusted party, Eve, who is expected to perform an interference measurement and record the detector ($L$ or $R$) that clicks.

\item
\textbf{Announcement}: Eve announces the detection result for each round. After that, Alice and Bob announce their random phases and intensity settings $\{j_a,\mu_a\}$ and $\{j_b,\mu_b\}$ and keep the signals with $\mu_a = \mu_b$.

\item
\textbf{Sifting}: When Eve announces a successful detection, (a click from exactly one of the detectors $L$ or $R$), Alice and Bob keep $\kappa_a$ and $\kappa_b$. Bob flips his key bit $\kappa_b$ if Eve's announcement was an $R$ click. Then, Alice and Bob group the signals by $j_s = (j_a - j_b) \textit{ mod } D$. If $j_s\in [\frac{D}{4}, \frac{3D}{4})$, Bob flips his key bit $\kappa_b$. After that, Alice and Bob merge the data with $j_s$ and $j_s + \frac{D}{2}$, with $j_s = 0,1,...,\frac{D}{2}-1$.

\item
\textbf{Parameter estimation}: For all the raw data that they have retained, Alice and Bob record the detect number $M_{i_a, i_b}^{(j_s)}$ of different intensity combinations $\{\mu_{i_a}, \mu_{i_b}\}$ and phase groups $j_s$. They then estimate the information leakage $E^{ph}_{\mu}$ using Eq.~\eqref{eq:Eph}.

\item
\textbf{Key generation}: For the signals with $\mu_{i_a} = \mu_{i_b} = \mu/2$, Alice and Bob group them by the phase index $j_s = \min[(j_a - j_b)\text{ mod } D, (j_b - j_a)\text{ mod } D]$. After that, they keep the signals with $j_s$ in a set $J$. They reconcile the corresponding key string to $\kappa_{A}$ by an encrypted classical channel, with the consummation of $l_{ec}^{j_s}$-bit keys according to $j_s$. After that, they perform privacy amplification according to the estimated single-photon ratio $q_1$ to generate keys.
\end{enumerate}

\section{Finite data size analysis} \label{Sc:finite}

In this section, we analyze the finite-size effect of PM-QKD protocol introduced in the main text. Here we ignore the effect caused by discrete phase randomization. Recall that during the post-processing in PM-QKD, Alice and Bob first reconcile their sifted raw key bits $K_a$ and $K_b$ with for each phase group $j_s = j_a - j_b$ separately. They keep some of the phase groups $j_s \in J$ for key generation. They then estimate the single photon fraction $q_{1}^J$ for all the remaining data in $J$ and perform privacy amplification.

As is mentioned in the maintext, in QKD finite-size analysis, one should take the cost and failure probability of channel authentication, error verification, privacy amplification, and parameter estimation into account. However, the cost of the first three steps are negligible comparing to the one in parameter estimation. When the final key length is much larger than $37$ bits, one can ignore the corresponding failure probability with a constant secret key cost \cite{Fung2010Practical}. For simplicity, we ignore these parts in our analysis. The finite-size key length $N_k$ is then
\begin{equation}
N_k = \sum_{j\in J} M_{si}^{j} [ 1 - H(E^{ph}) - fH(E^{j})],
\end{equation}
with the failure probability of $\epsilon_{eph}$. Here $M_{si}^j$ is the sifted raw key length of phase group $j$ before error verification and privacy amplification, $f$ is the error correction efficiency, $E^{j}$ is the quantum bit error rate of group $j$, $E^{ph}$ is the estimated upper bound of phase error rate bounded by
\begin{equation} \label{eq:ephMJ1}
E^{ph} < 1 - q_1^J = 1 - \frac{M^{\mu J}_1}{M^{\mu J}},
\end{equation}
with a failure probability of $\epsilon_{eph}$. Here $M^{\mu J}$ is the click round number with $\mu_a = \mu_b = \mu/2$ and phase group $j_s \in J$, and $M^{\mu J}_1$ is the estimated single photon component in it.

The core of parameter estimation of $E^{ph}$ is to estimate \emph{the clicked rounds caused by the single photon component $M_{1}^{\mu J}$, for all the clicked rounds with $\mu_a = \mu_b = \mu/2$ and all the left phase group $J$}. Here, we follow Ref.~\cite{zhang2017improved} for a tight decoy-state analysis in the finite-data regime.

As is stated in Appendix \ref{Sc:practical}, since the single photon state $\rho_1$ is the same whatever $\phi_a$ and $\phi_b$ is, we can first estimate the single photon clicks $M_{1}^{\mu}$ caused by all the phase groups with different $j_d$ all together, and then estimate $M_{1}^{\mu J}$ for the left groups $J$ afterwards.

Without matching the random phase $\phi_a, \phi_b$ by $j_s$, the state on optical modes $A,B$ sent out to Eve when $\mu_a = \mu_b = \mu/2$ is
\begin{equation}
\rho_{AB} = \sum_{k} P^\mu(k) \rho_k^{tot},
\end{equation}
where
\begin{equation}
\rho_k^{tol} = \int_0^\pi d\delta\, \rho_k^\delta =\sum_{k_a, k_b} B(k_a;k,\frac{1}{2}) \ket{k_a, k_b}\bra{k_a, k_b},
\end{equation}
and $B(k;n,p) = \binom{n}{k} p^k (1-p)^{(n-k)}$ is the binomial distribution. $\rho_k^\delta$ is given by Eq.~\eqref{eq:rhokdelta}. Note that $\rho_1^{tot} = \rho_1$. For different intensities $\{0,\nu,\mu\}$, the $k$-photon component $\rho_k^{tot}$ are the same, following Poisson distribution with different parameter. Therefore, the former finite-size decoy state analysis on BB84 protocol can be directly applied to PM-QKD.

For each intensity $\{0,\nu,\mu\}$, suppose Alice and Bob send out $N^{vac}, N^{w}, N^{s}$ rounds of pulses, with $M^{vac}, M^{w}, M^{s}$ rounds of effective clicks, respectively. We have,
\begin{equation}
\begin{aligned}
N^{vac} &\approx (r^{vac})^2 N, \\
N^{w} &\approx (r^{w})^2 N, \\
N^{s} &\approx (r^{s})^2 N.
\end{aligned}
\end{equation}
A unified notation is $\{N^{a}, M^{a}\}$, with $a\in\{vac,w,s\}$ indicating the intensity setting of vacuum, weak and signal. Denote the normalized rate of each intensity setting $a$ in the final clicked signals as
\begin{equation}
q^a = N^a/N,
\end{equation}
note that, $q^a$ is a fixed number after Alice and Bob send out all the signals.

For a specific intensity setting $a\in\{vac,w,s\}$, denote the rounds of sending out $k$-photon pulses $\rho^{tot}_k$ and clicks caused by it as $\{N^{a}_k, M^{a}_k\}$. Define
\begin{equation}
N_k \equiv \sum_a N^a_k,
\end{equation}
to be the overall rounds to send $k$-photon signals. Denote the conditional probability that Alice and Bob choose the intensity setting $a$ when sending out $k$-photon signal as
\begin{equation}
\begin{aligned}
p(a|k) &= \lim_{N_k \to \infty} \frac{N^a_k}{N_k} \\
&= \frac{ q^a P^{a}(k)}{ q^{vac} P^{0}(k) + q^w P^{\nu}(k) + q^s P^{\mu}(k) },
\end{aligned}
\end{equation}
here we slightly abuse the notation $P^a(k)$ to denote the Poisson distribution with intensity setting $a$.

Therefore we have
\begin{equation}
\begin{aligned}
N^a &= \sum_k N_k^a \approx \sum_k p(a|k)N_k \\
M^a &= \sum_k M_k^a \approx \sum_k p(a|k)M_k,
\end{aligned}
\end{equation}
where $p(vac|k)=0$ for all $k\neq0$.

In the whole QKD process, the values $\{N^a\}_a$ and $\{M^a\}_a$ are known by Alice and Bob. The values $\{N^a_k\}_{a,k}$ are available to Alice and Bob in principle and cannot be controlled by Eve. The values $\{M_k\}_k$ are, however, controlled by Eve and unknown to Alice and Bob.

The core observations to perform finites-size analysis in this case are
\begin{enumerate}
\item When Alice and Bob send out $k$-photon signals, the choosing of different intensity setting $a$ is independent and identically distributed (i.i.d.), given by the probability distribution $p(a|k)$.
\item Eve's attack on $k$-photon signals cannot depend on the intensity setting $a$.
\end{enumerate}

Therefore, Eve's attack can be described by a random sampling from the set of $N_k$. She randomly chooses $M_k$ rounds from it and announces them as effective clicks. Among them, $\{M^a_k\}_a$, i.e., the clicks caused by different intensity settings, are randomly distributed.

To clarify the random sampling model, we can then rewrite $M^a_k$ as
\begin{equation} \label{eq:barMak}
\bar{M}^a_k = \sum_{i=1}^{M_k} (\bar{\chi}^a_k)_i,
\end{equation}
where
\begin{equation}
(\bar{\chi}^a_k)_i =
\begin{cases}
1 & \text{with probability } p^a_k, \\
0 & \text{with probability } 1 - p^a_k,
\end{cases}
\end{equation}
(with $i=1,...,M_k$) are i.i.d. indicator random variables.

Group these random variables and define
\begin{equation} \label{eq:barMa}
\bar{M}^a = \sum_{k} \bar{M}^a_k = \sum_{k} \sum_{i=1}^{M_k} (\bar{\chi}^a_k)_i,
\end{equation}
as the variable indicating the overall clicks caused by intensity setting $a$. The bar on $M^a_k$ is used to indicate that it is a variable.

The decoy state problem can be modeled as: for the unknown $\{M_k\}_k$ and the known $\{p^a_k\}_{a,k}$, to evaluate the value of the variable $\bar{M}^s_1$, given the observed constraints that $\{\bar{M}^a = M^a\}_a$.

We first observe that,
\begin{equation}
\mathbb{E}(\bar{M}^a_k) = p(a|k) M_k,
\end{equation}
and hence
\begin{equation} \label{eq:EMak}
\mathbb{E}(\bar{M}^a) = \sum_k \mathbb{E}(\bar{M}^a_k) = \sum_k p(a|k) M_k.
\end{equation}
Note that, \emph{the expectation value are taken with respect to the i.i.d. variables $\{(\bar{\chi}^a_i)_j\}_{a,i,j}$}. Therefore, we can bound the expected values $\mathbb{E}(\bar{M}^a)$ by applying an inversed form of Chernoff bound on Eq.~\eqref{eq:barMa} and with the observed $\{M^a\}_a$. From Eq.~\eqref{eq:EMak}, we have

\begin{equation} \label{eq:decoyrange}
\mathbb{E}^U(\bar{M}^a) \geq  \sum_k p(a|k)M_k \geq \mathbb{E}^L(\bar{M}^a),
\end{equation}
where we use superscript U and L to denote upper and lower bound respectively.

To estimate $M^s_1$, we first estimate $M_1$ from Eq.~\eqref{eq:decoyrange}, and then estimate $M^s_1$ by a direct use of Chernoff bound on Eq.~\eqref{eq:barMak}. We briefly summarize the results of Chernoff bound in Appendix \ref{eq:Chernoff}.

The decoy method discussed above is based on the formula Eq.~\eqref{eq:EMak} and the correlation between variables $\{\bar{M}^a,M_k\}$. To unify it with the former decoy state formulas with $\{Q^a,Y_k\}$, we further define the gain and yield variable as
\begin{equation} \label{eq:QaYk}
\begin{aligned}
\bar{Q}^a &:= \frac{\bar{M}^a}{N^a},\\
Y^*_k &:= \frac{M_k}{N^\infty_k},
\end{aligned}
\end{equation}
where
\begin{equation}
\begin{aligned}
N_k^\infty &= \sum_a P^a(k) N^a = \sum_a P^a(k) (r^a)^2 N,\\
\end{aligned}
\end{equation}
is the expectation value of $N_k$.

From Eq.~\eqref{eq:EMak} and the definition of $\bar{Q}^a, \bar{Y}^*_k$ in Eq.~\eqref{eq:QaYk}, we can recover the decoy state formula expressed by $\bar{Q}^a$ and $\bar{Y}^*_k$,
\begin{widetext}
\begin{equation}
\begin{aligned}
\mathbb{E}[\bar{Q}^a] = \frac{\mathbb{E}[\bar{M}^a]}{N^a} &= \frac{\sum_k p(a|k) M_k}{N^a} = \sum_k \frac{ q^a P^{a}(k)}{ q^{vac} P^{0}(k) + q^w P^{\nu}(k) + q^s P^{\mu}(k) } \frac{M_k}{q^a N}, \\
&= \sum_k P^a(k) \frac{M_k}{N [ q^{vac} P^{0}(k) + q^w P^{\nu}(k) + q^s P^{\mu}(k) ]}, \\
&= \sum_k P^a(k) \frac{M_k}{N^\infty_k}, \\
&= \sum_{k} P^a(k) Y^*_k.
\end{aligned}
\end{equation}
\end{widetext}

Now, with the observed value $M^a$, we can calculate $Q^a$, and apply the decoy state formulas,

\begin{equation}  \label{eq:decoyQY}
\mathbb{E}^U[\bar{Q}^a] \geq \sum_k P^\mu(k)Y^*_k \geq \mathbb{E}^L[\bar{Q}^a],
\end{equation}

With \eqref{eq:decoyQY}, we can estimate $Y^*_1$ by\cite{zhang2017improved}
\begin{widetext}
\begin{equation} \label{eq:Y1star}
Y^*_1 \geq (Y_1^*)^L = \frac{\mu}{\mu\nu - \nu^2} \left( \mathbb{E}^L[Q^w]e^\nu - \mathbb{E}^U[Q^s]e^\mu \frac{\nu^2}{\mu^2} - \frac{\mu^2 - \nu^2}{\mu^2} \mathbb{E}^U[Q^{vac}] \right).
\end{equation}
\end{widetext}

Note that the whole process can be divided into two steps. Step I is to estimate $(Y_1^*)^L$ and $(M_1)^L$ for all of the phase groups. Step II is to estimate $(M_1^{s,J})^L$ for phase group set $J$ from $(M_1)^L$ in all phase group. To summarize, the whole phase error estimation process is
\begin{enumerate}
\item To record $\{N^s, N^w, N^{vac}\}$ and record the number of clicked rounds $\{M^s, M^w, M^{vac}\}$.
\item (Step I) Based on an inversed usage of Chernoff bound, to calculate $\{ \mathbb{E}^U(\bar{M}^a), \mathbb{E}^L(\bar{M}^a) \}_a$, given $M^a$ and estimate the failure probability $\epsilon_1$. Calculate the $\{ \mathbb{E}^U(\bar{Q}^a), \mathbb{E}^L(\bar{Q}^a) \}_a$ by Eq.~\eqref{eq:QaYk}.
\item Calculate the lower bound $(Y_1^*)^L$ based on $\{ \mathbb{E}^U(\bar{Q}^a), \mathbb{E}^L(\bar{Q}^a) \}_a$ by Eq.~\eqref{eq:Y1star}. Calculate $(M_1)^L$ by Eq.~\eqref{eq:QaYk}.
\item (Step II) Based on a direct usage of Chernoff bound, to calculate $(M^{s,J}_1)^L$ for phase group set $J$ and estimate the failure probability $\epsilon_2$. To calculate $E^{ph}$ based on Eq.~\eqref{eq:ephMJ1}. The overall failure probability is $\epsilon_{eph} = \epsilon_1 + \epsilon_2$.
\end{enumerate}

\subsection{Chernoff bound} \label{eq:Chernoff}

Here we present the methods to evaluate $\mathbb{E}(\bar{M}^a)$ from $M^a$ and evaluate $M^s_1$ from $M_1$ using Chernoff bounds.

To evaluate $\mathbb{E}(\bar{M}^a)$ from $M^a$, we inversely use the Chernoff bounds based on Bernoulli variables. We briefly summarize the results in Ref.~\cite{zhang2017improved}. For the observed value $\chi$, we set the lower and upper bound of estimated $\mathbb{E}(\chi)$ as $\{\mathbb{E}^L(\chi), \mathbb{E}^U(\chi)\}$. Denote
\begin{equation}
\begin{aligned}
\mathbb{E}^L(\chi) &= \frac{\chi}{1 + \delta^L}, \\
\mathbb{E}^U(\chi) &= \frac{\chi}{1 - \delta^U}. \\
\end{aligned}
\end{equation}
The failure probability of the estimation $\mathbb{E}(\chi) \in [\mathbb{E}^L(\chi), \mathbb{E}^U(\chi)]$, given by the Chernoff bound, is
\begin{equation} \label{eq:eps1}
\epsilon = e^{-\chi g_2(\delta^L)} + e^{-\chi g_2(\delta^U)},
\end{equation}
where $g_2(x) = \ln(1+x) - x/(1+x)$.

To evaluate $M^s_1$ from $M_1$, we directly apply the Chernoff bounds. Suppose the direct sampling expectation value of $M^s_1$ is given by $\mathbb{E}(M^s_1) = p_1^s M_1$. For the expected value $\mathbb{E}(\chi)$, we set the lower and upper bound of the estimated $\chi$ as $\{\chi^L,\chi^U\}$. Denote
\begin{equation}
\begin{aligned}
\chi^L &= (1 - \bar{\delta}^L)\mathbb{E}(\chi), \\
\chi^U &= (1 + \bar{\delta}^U)\mathbb{E}(\chi). \\
\end{aligned}
\end{equation}
The failure probability of the estimation $\chi \in [\chi^L, \chi^U]$, given by the Chernoff bound, is
\begin{equation} \label{eq:eps2}
\epsilon = e^{-(\bar{\delta}^L)^2\mathbb{E}(\chi)/(2+\bar{\delta}^L)} + e^{-(\bar{\delta}^U)^2\mathbb{E}(\chi)/(2+\bar{\delta}^U)}.
\end{equation}

In practice, we can preset the lower bound and upper bound $\{\mathbb{E}^L(\chi), \mathbb{E}^U(\chi)\}$ or $\{\chi^L,\chi^U\}$ by assuming a Gaussian distribution on $\chi$ first,
\begin{equation}
\begin{aligned}
\mathbb{E}^L(\chi) = \chi - n_\alpha \sqrt{\chi},&\quad \mathbb{E}^U(\chi) = \chi + n_\alpha \sqrt{\chi}, \\
\chi^L = \chi - n_\alpha \sqrt{\mathbb{E}(\chi)},&\quad \chi^U = \chi + n_\alpha \sqrt{\mathbb{E}(\chi)}, \\
\end{aligned}
\end{equation}
where $n_\alpha$ is a preset parameter to determine the estimation precision. After that, we calculate the failure probabilities by Eq.~\eqref{eq:eps1} and Eq.~\eqref{eq:eps2}.

\section{Simulation formula and results} \label{Sc:simulationformula}
Here we list the formulas used to simulate the key rate performance of PM-QKD and MDI-QKD in Fig.~2 the main text. The channel is modeled to be a pure loss one and symmetric for Alice and Bob with transmittance $\eta$ (with the detector efficiency $\eta_d$ taken into account).

\subsection{Gain, yield and error rate of PM-QKD}
In PM-QKD, suppose Alice and Bob emit the $k$-photon light $\rho^\delta_k$ in Eq.~\eqref{eq:rhokdelta}, the yield (i.e., effective detection probability) $Y_k$ is given by, Eq.~(B13) in Ref.~\cite{ma2018phase},
\begin{equation}
Y_k \approx 1 - (1 - 2 p_d) (1-\eta)^k,
\end{equation}
suppose Alice and Bob emit the coherent states $\rho$ in Eq.~\eqref{eq:muamub0pi} with $\mu_a = \mu_b = \mu/2$, the gain (i.e., effective detection probability) $Q_{\mu}$ is ( Eq.~(B14) in Ref.~\cite{ma2018phase})
\begin{equation}
Q_\mu \approx 1 - (1 - 2 p_d) e^{-\eta \mu}.
\end{equation}

As is stated in the main text, the quantum bit error rate $E^{j_s}$ is mainly composed of three components. The first one is the intrinsic error $e_\Delta(j_s)$ caused by phase mismatch when $j_s \neq 0$,
\begin{equation}
e_\Delta(j_s) =
\begin{cases}
\sin^2(\frac{\pi j_s}{D}), & j_s \leq \frac{D}{2}, \\
\sin^2(\frac{\pi}{2} - \frac{\pi j_s}{D}), & j_s > \frac{D}{2}, \\
\end{cases}
\end{equation}
which is related to the index deviation $j_s$. The second one is the extra misalignment error $e_0$, caused by phase fluctuation. Here we regard $e_0$ and $e_\Delta(j_s)$ as caused by independent factors, and the overall misalignment error is $e_d(j_s) = e_0 + e_\Delta(j_s)$. Also, the dark count effect will contribute to the bit error. The overall bit error rate $E_\mu^{(j_s)}$ is then given by
\begin{equation}
E_\mu^{(j_s)} = [p_d + \eta\mu (e_\Delta(j_s) + e_0)]\frac{e^{-\eta\mu}}{Q_\mu},
\end{equation}
where $p_d$ is the dark count rate.

\subsection{Simulation formulas for MDI-QKD protocols} \label{Sc:KBB84MDI}

The key rate of MDI-QKD is \cite{Lo2012Measurement}
\begin{equation}
\begin{aligned}
R_{MDI} = \dfrac{1}{2}\{ Q_{11} [1 - H(e_{11})] -f Q_{rect} H(E_{rect}) \},
\end{aligned}
\end{equation}
where $Q_{11} = \mu_a\mu_b e^{-\mu_a-\mu_b}Y_{11}$ and $1/2$ is the basis sifting factor. We take this formula from Eq.~(B27) in Ref.~\cite{Ma2012Alternative}. In simulation, the gain and error rates are
\begin{widetext}
\begin{equation}
\begin{aligned}
Y_{11} & = (1-p_d)^2 [ \dfrac{\eta_a\eta_b}{2} + (2\eta_a + 2\eta_b -3\eta_a\eta_b)p_d + 4(1-\eta_a)(1-\eta_b)p_d^2 ], \\
e_{11} & = e_0Y_{11} - (e_0- e_d)(1-p_d^2)\dfrac{\eta_a\eta_b}{2}, \\
Q_{rect} & = Q_{rect}^{(C)} + Q_{rect}^{(E)}, \\
Q_{rect}^{(C)} & = 2(1-p_d)^2 e^{-\mu^\prime/2} [1 - (1-p_d)e^{-\eta_a\mu_a/2}][1 - (1-p_d)e^{-\eta_b\mu_b/2}], \\
Q_{rect}^{(E)} & = 2p_d(1-p_d)^2 e^{-\mu^\prime/2}[I_0(2x) - (1-p_d)e^{-\mu^\prime/2}], \\
E_{rect} Q_{rect} & = e_d Q_{rect}^{(C)} + (1 - e_d) Q_{rect}^{(E)}, \\
\end{aligned}
\end{equation}
\end{widetext}
where $\mu^\prime$ denotes the average number of photons reaching Eve's beam splitter,   $\mu_a = \mu_b = \mu/2, \eta_a = \eta_b = \eta$, and
\begin{equation} \label{eqn:MDI mu x}
\begin{aligned}
\mu^\prime & = \eta_a \mu_a + \eta_b \mu_b, \\
x & = \frac12\sqrt{\eta_a \mu_a \eta_b \mu_b}. \\
\end{aligned}
\end{equation}
We take these formulas from Eqs.~(A9),~(A11),~(B7),~and (B28)-(B31) in Ref.~\cite{Ma2012Alternative}.

The linear key-rate bound of repeaterless point-to-point QKD protocol used in the maintext is \cite{Pirandola2017Fundamental},
\begin{equation}\label{eq:plob}
R_{PLOB} = - \log_2(1-\eta).
\end{equation}

\end{appendix}

\bibliographystyle{apsrev4-1}
\bibliography{bibPMPlus} 

\end{document}